\documentclass[transmag]{IEEEtran}
\usepackage{latexsym}
\usepackage{graphicx}
\usepackage{amsfonts,amssymb,amsmath}
\usepackage{hyperref}
\usepackage{graphicx}
\usepackage{bm}
\usepackage{enumerate}
\usepackage{float}
\usepackage{multicol}

\usepackage{balance}

\usepackage{color}
\usepackage{url}

\usepackage{cite}

\usepackage{amsmath}
\usepackage{amsthm}

\usepackage{amssymb}

\DeclareMathOperator*{\argmax}{argmax}
 \newtheorem{lem}{Lemma}

\newtheorem{corollary}{Corollary}

\usepackage{algorithmic}
 \usepackage{algorithm}

\usepackage{array}

\def\BibTeX{{\rm B\kern-.05em{\sc i\kern-.025em b}\kern-.08em T\kern-.1667em\lower.7ex\hbox{E}\kern-.125emX}}
\begin{document}

\title{A Semi-Linear Approximation of the First-Order Marcum $Q$-function with Application to Predictor Antenna Systems}

\author{Hao~Guo,~\IEEEmembership{Student~Member,~IEEE},
        Behrooz~Makki,~\IEEEmembership{Senior~Member,~IEEE},\\
        Mohamed-Slim Alouini,~\IEEEmembership{Fellow,~IEEE},
        and Tommy~Svensson,~\IEEEmembership{Senior~Member,~IEEE}
\thanks{H. Guo and T. Svensson are with the Department of Electrical Engineering, Chalmers University of Technology, 41296 Gothenburg, Sweden (email: hao.guo@chalmers.se; tommy.svensson@chalmers.se).}
\thanks{B. Makki is with Ericsson Research, 41756 Gothenburg, Sweden (email: behrooz.makki@ericsson.com).}
\thanks{M.-S. Alouini is with the King Abdullah University of Science and Technology,
Thuwal 23955-6900, Saudi Arabia (e-mail: slim.alouini@kaust.edu.sa).}}

\IEEEtitleabstractindextext{\begin{abstract}
First-order Marcum $Q$-function is observed in various problem formulations. However, it is not an easy-to-handle function. For this reason, in this paper, we first present a semi-linear approximation of the Marcum $Q$-function. Our proposed approximation is useful because it simplifies, e.g., various integral calculations including Marcum $Q$-function as well as different operations such as parameter optimization. Then, as an example of interest, we apply our proposed approximation approach to the performance analysis of predictor antenna (PA) systems. Here, the PA system is referred to as a system with two sets of antennas on the roof of a vehicle. Then, the PA positioned in the front of the vehicle can be used to improve the channel state estimation for data transmission of the receive antenna  that is aligned behind the PA. Considering spatial mismatch due to the mobility, we derive closed-form expressions for the instantaneous and average throughput as well as the throughput-optimized rate allocation. As we show, our proposed approximation scheme enables us to analyze PA systems with high accuracy. Moreover, our results show that rate adaptation can improve the performance of PA systems with different levels of spatial mismatch.
\end{abstract}

}

\maketitle

\section{Introduction}

The first-order Marcum $Q$-function\footnote{To simplify the analysis, our paper concentrates on the approximation of the first-order Marcum-$Q$ function. However, our approximation technique can be easily extended to the cases with different orders of Marcum $Q$-function.} is defined as  \cite[Eq. (1)]{Bocus2013CLapproximation}
\begin{align}\label{eq_Qdefi}
    Q_1(\alpha,\beta) = \int_{\beta}^{\infty} xe^{-\frac{x^2+\alpha^2}{2}}I_0(x\alpha)\text{d}x,
\end{align}
where $\alpha, \beta \geq 0$ and $I_n(x) = (\frac{x}{2})^n \sum_{i=0}^{\infty}\frac{(\frac{x}{2})^{2i} }{i!\Gamma(n+i+1)}$ is the $n$-order modified Bessel function of the first kind, and $\Gamma(z) = \int_0^{\infty} x^{z-1}e^{-x} \mathrm{d}x$ represents the Gamma function. Reviewing the literature,  Marcum $Q$-function has appeared in many areas such as statistics/signal detection \cite{helstrom1994elements}, and in the performance analysis of different setups such as temporally correlated channels \cite{Makki2013TCfeedback}, spatially correlated channels \cite{Makki2011Eurasipcapacity},  free-space optical (FSO) links \cite{Makki2018WCLwireless}, relay networks \cite{Makki2016TVTperformance}, as well as  cognitive radio and radar systems \cite{Simon2003TWCsome,Suraweera2010TVTcapacity,Kang2003JSAClargest,Chen2004TCdistribution,Ma2000JSACunified,Zhang2002TCgeneral,Ghasemi2008ICMspectrum,Digham2007TCenergy, simon2002bookdigital,Cao2016CLsolutions,sofotasios2015solutions,Cui2012ELtwo,Azari2018TCultra,Alam2014INFOCOMWrobust,Gao2018IAadmm,Shen2018TVToutage,Song2017JLTimpact,Tang2019IAan,ermolova2014laplace,peppas2013performance}.

The presence of Marcum $Q$-function, however, makes the mathematical analysis challenging, because it is  difficult to manipulate  with no closed-form  expressions especially when it appears in parameter optimizations and integral calculations. For this reason, several methods have been developed in \cite{Bocus2013CLapproximation,Fu2011GLOBECOMexponential,zhao2008ELtight,Simon2000TCexponential,annamalai2001WCMCcauchy,Sofotasios2010ISWCSnovel,Li2010TCnew,andras2011Mathematicageneralized,Gaur2003TVTsome,Kam2008TCcomputing,Corazza2002TITnew,Baricz2009TITnew,chiani1999ELintegral,jimenez2014connection}  to bound/approximate the  Marcum $Q$-function. For example, \cite{Fu2011GLOBECOMexponential,zhao2008ELtight} have proposed  modified forms of the function, while \cite{Simon2000TCexponential,annamalai2001WCMCcauchy} have derived exponential-type bounds which are good for  bit error rate analysis at high signal-to-noise ratios (SNRs). Other types of bounds are expressed by, e.g., error function \cite{Kam2008TCcomputing} and Bessel functions \cite{Corazza2002TITnew,Baricz2009TITnew,chiani1999ELintegral}. Some alternative methods have been also developed in \cite{Sofotasios2010ISWCSnovel,Li2010TCnew,andras2011Mathematicageneralized,Bocus2013CLapproximation,Gaur2003TVTsome}. Although each of these approximation/bounding techniques are fairly tight for their considered problem formulation, they are still based on difficult functions, or have complicated summation/integration formations, which may be not easy to deal with in, e.g., integral calculations and parameter optimizations.

In this paper, we propose a simple and semi-linear approximation for Marcum Q-function, and present an application of the developed approximation in improving the backhaul performance of moving relays. Particularly, the contributions of the paper are two-fold as highlighted in the following.

\subsection{Semi-linear Approximation of Marcum Q-function}
We first propose a simple semi-linear approximation of the first-order Marcum $Q$-function (Lemma \ref{Lemma1}, Corollaries \ref{coro1}-\ref{coro2}). As we explain in the following (Lemmas \ref{Lemma1}-\ref{Lemma4}), in contrast to the schemes of \cite{Bocus2013CLapproximation,Fu2011GLOBECOMexponential,zhao2008ELtight,Simon2000TCexponential,annamalai2001WCMCcauchy,Sofotasios2010ISWCSnovel,Li2010TCnew,andras2011Mathematicageneralized,Gaur2003TVTsome,Kam2008TCcomputing,Corazza2002TITnew,Baricz2009TITnew,chiani1999ELintegral,jimenez2014connection}, our proposed approximation is not tight at the tails of  the Marcum $Q$-function because we use simple straight line(s) to approximate a curve. Therefore, it is not useful in, e.g., error probability-based problem formulations.  On the other hand, the advantages of our proposed approximation method, compared to \cite{Bocus2013CLapproximation,Fu2011GLOBECOMexponential,zhao2008ELtight,Simon2000TCexponential,annamalai2001WCMCcauchy,Sofotasios2010ISWCSnovel,Li2010TCnew,andras2011Mathematicageneralized,Gaur2003TVTsome,Kam2008TCcomputing,Corazza2002TITnew,Baricz2009TITnew,chiani1999ELintegral,jimenez2014connection}, are 1) its simplicity, and 2) tightness in the moderate values of the function. This is important because, as observed in, e.g., \cite{Bocus2013CLapproximation,Fu2011GLOBECOMexponential,Makki2013TCfeedback,Makki2011Eurasipcapacity,Makki2018WCLwireless,Makki2016TVTperformance,Simon2003TWCsome,Suraweera2010TVTcapacity,Ma2000JSACunified,Digham2007TCenergy,Cao2016CLsolutions,sofotasios2015solutions,Azari2018TCultra,Alam2014INFOCOMWrobust,Gao2018IAadmm,Shen2018TVToutage,Song2017JLTimpact,Tang2019IAan}, in different applications,  Marcum $Q$-function is typically combined with other functions which tend to zero at the tails of the Marcum $Q$-function. In such cases, the inaccuracy of the approximation at the tails does not affect the tightness of the final analysis. Thus, our proposed scheme provides tight and simple approximation results for different problem formulations such as capacity calculation \cite{Makki2011Eurasipcapacity,Suraweera2010TVTcapacity}, throughput/average rate derivation \cite{Makki2013TCfeedback,Makki2018WCLwireless,Makki2016TVTperformance}, energy detection of unknown signals over various multipath fading channels \cite{Digham2007TCenergy,Cao2016CLsolutions,sofotasios2015solutions}, as well as  performance evaluation of  non-coherent receivers in radar systems \cite{Cui2012ELtwo} (Lemmas \ref{Lemma2}-\ref{Lemma3}). Also, the simplicity of the approximation method makes it possible to perform further analysis such as parameter optimization and to obtain intuitive insights from the derivations.

\subsection{Throughput Optimization of Predictor Antenna Systems} 
To demonstrate the usefulness of the proposed approximation
technique in communication systems, we analyze the
performance of predictor antenna (PA) systems in presence of
spatial mismatch. Here, the PA system is referred to as a setup
with two (sets of) antennas on the roof of a vehicle. The PA
positioned in the front of a vehicle can be used to improve
the channel state estimation for downlink data reception at the
receive antenna (RA) on the vehicle that is aligned behind the
PA \cite{Sternad2012WCNCWusing,DT2015ITSMmaking,BJ2017PIMRCpredictor,phan2018WSAadaptive,Jamaly2014EuCAPanalysis, BJ2017ICCWusing,Apelfrojd2018PIMRCkalman,Jamaly2019IETeffects,Guo2019WCLrate,guo2020power,guo2020rate,guo2020commnet}.

The feasibility of PA setups, which are of interest particularly in public transport systems such as trains and buses, but potentially also for the more design-constrained cars, has been previously shown through experimental tests \cite{Sternad2012WCNCWusing,DT2015ITSMmaking,BJ2017PIMRCpredictor,phan2018WSAadaptive,Jamaly2014EuCAPanalysis, BJ2017ICCWusing,Apelfrojd2018PIMRCkalman}. Particularly, as shown in testbed implementations, e.g.,  \cite{BJ2017PIMRCpredictor,BJ2017ICCWusing}, with a two-antenna PA setup a normalised mean square error  of around -10 dB can be obtained for speeds up to 50 km/h, with  measured predictions horizons up to three times the wavelengths. This is by an order of magnitude better than state-of-the-art Kalman prediction-based systems \cite{Ekman2002,Aronsson2011} with prediction horizon limited to 0.1-0.3 times the wavelength. Moreover, the  European project deliverables, e.g.,  \cite[Chapter 2]{ARTIST4G},  \cite[P. 107]{metis2015d33} and \cite[Chapter 3]{5gcar2019d33}, have  well addressed the feasibility of the PA concept in network-level design. Finally, different works have analyzed the PA system in both frequency division duplex (FDD) \cite{BJ2017ICCWusing,BJ2017PIMRCpredictor,phan2018WSAadaptive} and time division duplex (TDD) \cite{DT2015ITSMmaking,Apelfrojd2018PIMRCkalman} systems, with some developments on addressing the system challenges such as  antenna coupling \cite{Jamaly2014EuCAPanalysis,Jamaly2019IETeffects}, spatial mismatch \cite{Jamaly2019IETeffects,Guo2019WCLrate}, and spectrum underutilization \cite{guo2020power,guo2020rate}. 

Among the challenges of the PA system is the spatial mismatch. If the RA does not arrive in the same position as the PA, the actual channel for the RA would not be identical to the one experienced by the PA before. Such inaccurate channel state information (CSI) estimation will affect the system performance considerably at moderate/high speeds \cite{DT2015ITSMmaking,Guo2019WCLrate}. One way to compensate for this mismatch problem is by interpolating channel samples at the base station (BS) \cite{DT2015ITSMmaking,BJ2017PIMRCpredictor, Apelfrojd2018PIMRCkalman}. Using all RAs as PA in the first time slot, \cite{DT2015ITSMmaking} proposes an interpolation-based scheme for multi-antenna systems with beamforming. Then, \cite{BJ2017PIMRCpredictor} shows that an important part of the implementation of the PA is to identify the sample that is closest in space, and if necessary, interpolate between previous estimates. Also, in \cite{Apelfrojd2018PIMRCkalman} Kalman smoother is used to interpolate larger distances that might occur in TDD. However, interpolation comes with additional overhead due to the increased number of pilots/PAs, and this overhead becomes more severe under FDD setup. This may lead to performance losses, in terms of end-to-end throughput. Also, studies in \cite{DT2015ITSMmaking,BJ2017PIMRCpredictor, Apelfrojd2018PIMRCkalman} have not included analytical analysis for the PA system.

In this paper, we address the spatial  mismatch problem by implementing adaptive rate allocation without increasing the number of pilots/PAs. In our proposed setup, the instantaneous CSI provided by the PA is used to adapt the data rate of the signals sent to the RA from the BS. The problem is cast in the form of throughput maximization. Particularly, we use our developed approximation approach (Lemma \ref{Lemma1}, Corollaries \ref{coro1}-\ref{coro2}) to derive closed-form expressions for the instantaneous and average throughput as well as the optimal rate allocation maximizing the throughput (Lemma \ref{Lemma4}). Moreover, we study the effect of different parameters such as the antennas distance, the vehicle speed, and the processing delay of the BS on the performance of PA setups.

Our paper is different from the state-of-the-art literature because the proposed semi-linear approximation of the first-order Marcum $Q$-function and the derived closed-form expressions for the considered integrals, i.e., Lemmas \ref{Lemma1}-\ref{Lemma4} and Corollaries \ref{coro1}-\ref{coro2}, have not been presented by, e.g., \cite{Bocus2013CLapproximation,helstrom1994elements,Fu2011GLOBECOMexponential,Makki2013TCfeedback,Makki2011Eurasipcapacity,Makki2018WCLwireless,Makki2016TVTperformance,Simon2003TWCsome,Suraweera2010TVTcapacity,Kang2003JSAClargest,Chen2004TCdistribution,Ma2000JSACunified,Zhang2002TCgeneral,Ghasemi2008ICMspectrum,Digham2007TCenergy, simon2002bookdigital,Cao2016CLsolutions,sofotasios2015solutions,Cui2012ELtwo,Azari2018TCultra,Alam2014INFOCOMWrobust,Gao2018IAadmm,Shen2018TVToutage,Song2017JLTimpact,Tang2019IAan,ermolova2014laplace,peppas2013performance,zhao2008ELtight,Simon2000TCexponential,annamalai2001WCMCcauchy,Sofotasios2010ISWCSnovel,Li2010TCnew,andras2011Mathematicageneralized,Gaur2003TVTsome,Kam2008TCcomputing,Corazza2002TITnew,Baricz2009TITnew,chiani1999ELintegral,jimenez2014connection,Sternad2012WCNCWusing,DT2015ITSMmaking,BJ2017PIMRCpredictor,phan2018WSAadaptive,Jamaly2014EuCAPanalysis, BJ2017ICCWusing,Apelfrojd2018PIMRCkalman,Jamaly2019IETeffects,Guo2019WCLrate}. Also, as opposed to \cite{Sternad2012WCNCWusing,DT2015ITSMmaking,BJ2017PIMRCpredictor,phan2018WSAadaptive,Jamaly2014EuCAPanalysis, BJ2017ICCWusing,Apelfrojd2018PIMRCkalman,Jamaly2019IETeffects}, we perform analytical performance evaluation of PA systems with CSIT (T: at the transmitter)-based rate optimization to mitigate the effect of the spatial mismatch. Moreover, compared to our preliminary results in \cite{Guo2019WCLrate}, this paper develops the semi-linear approximation method for the Marcum $Q$-function, and uses our proposed approximation method to analyze the performance of the PA system. Also, we perform deep analysis of the effect of various parameters, such as imperfect CSIT feedback schemes, and processing delay of the BS on the system performance.

The simulation and the analytical results indicate that the proposed semi-linear approximation is useful for the mathematical analysis of different Marcum $Q$-function-based problem formulations. Particularly, our approximation method enables us to represent different Marcum $Q$-function-based integrations and optimizations in closed-form. Considering the PA system, our derived analytical results show that adaptive rate allocation can considerably improve the performance of the PA system in the presence of spatial mismatch. Finally, with different levels of channel estimation, our results show that there exists an optimal speed for the vehicle optimizing the throughput/outage probability, and the system performance is sensitive to the vehicle speed/processing delay as the speed moves away from its optimal value.

This paper is organized as follows. In Section \ref{sec:appro}, we present our proposed semi-linear approximation of the first-order Marcum $Q$-function, and derive closed-form solutions for some integrals of interest. Section \ref{sec:application} deals with the application of the approximation in the PA system,  deriving closed-form expressions for the optimal rate adaptation, the instantaneous throughput as well as the expected throughput.  In this way, Sections \ref{sec:appro} and \ref{sec:application} demonstrate examples on how the proposed approximation can be useful in, respectively, expectation- and optimization-based problem formulations involving the Marcum $Q$-function, respectively. Concluding remarks are provided in Section \ref{sec:conclu}.

\section{Approximation of the first-order Marcum Q-function}\label{sec:appro}
In this section, we present our semi-linear approximation of the cumulative distribution function (CDF) of   the form  $y(\alpha, \beta ) = 1-Q_1(\alpha,\beta)$, with Marcum Q function  $Q_1(\alpha,\beta)$ defined in (\ref{eq_Qdefi}). The idea of this proposed approximation is to use one point and its corresponding slope in that point to create a line  approximating the CDF. The approximation method is summarized in Lemma \ref{Lemma1} as follows.
\begin{lem}\label{Lemma1}
 The CDF of the form $y(\alpha, \beta ) = 1-Q_1(\alpha,\beta)$ can be semi-linearly approximated as $y(\alpha,\beta)\simeq \mathcal{Z}(\alpha, \beta)$ where
\begin{align}\label{eq_lema1}
\mathcal{Z}(\alpha, \beta)=
\begin{cases}
0,  ~~~~~~~~~~~~~~~~~~~~~~~~~~~~~\mathrm{if}~ \beta < c_1  \\ 
 \beta_0 e^{-\frac{1}{2}\left(\alpha^2+\left(\beta_0\right)^2\right)}I_0\left(\alpha\beta_0\right)\left(\beta-\beta_0\right)+\\~~~1-Q_1\left(\alpha,\beta_0\right),  ~~~~~~~~~\mathrm{if}~ c_1 \leq\beta\leq c_2 \\
1,  ~~~~~~~~~~~~~~~~~~~~~~~~~~~~~\mathrm{if}~ \beta> c_2,
\end{cases}
\end{align}
with
\begin{align}\label{eq_alpha0}
   \beta_0 =  \frac{\alpha+\sqrt{\alpha^2+2}}{2},
\end{align}
\begin{align}\label{eq_c1}
   c_1(\alpha) = \max\Bigg(0,\beta_0+\frac{Q_1\left(\alpha,\beta_0\right)-1}{\beta_0 e^{-\frac{1}{2}\left(\alpha^2+\left(\beta_0\right)^2\right)}I_0\left(\alpha\beta_0\right)}\Bigg),
\end{align}
and
\begin{align}\label{eq_c2}
   c_2(\alpha) = \beta_0+\frac{Q_1\left(\alpha,\beta_0\right)}{\beta_0 e^{-\frac{1}{2}\left(\alpha^2+\left(\beta_0\right)^2\right)}I_0\left(\alpha\beta_0\right)}.
\end{align}
\end{lem}
\begin{proof}

We aim to approximate the CDF in the range $y \in [0, 1]$ with respect to $\beta$ by 
\begin{align}\label{eq_YY}
    y-y_0 = m(\beta-\beta_0),
\end{align}
where $\mathcal{C} = (\beta_0,y_0)$ is a point on the CDF curve and $m$ is the slope  of $y(\alpha,\beta)$ at point $\mathcal{C}$. Then, the parts of the line outside this region are replaced by $y=0$ and $y=1$ (see Fig. \ref{fig_CDFilu}).

To obtain a good approximation of the CDF, we select the point $\mathcal{C}$ by finding the steepest slope through solving
\begin{align}\label{eq_partialsquare}
   \beta_0 = \mathop{\arg}_{t} \left\{  \frac{\partial^2\left(1-Q_1(\alpha,t)\right)}{\partial t^2} = 0\right\},
\end{align}
because the Marcum $Q$-function is symmetric around this point so a linear function through point $\mathcal{C}$ gives the best fit. Then, using the derivative of the first-order Marcum $Q$-function with respect to $t$ \cite[Eq. (2)]{Pratt1968PIpartial}
\begin{align}\label{eq_derivativeMarcumQ}
    \frac{\partial Q_1(\alpha,t)}{\partial t} = -t e^{-\frac{\alpha^2+t^2}{2}}I_0(\alpha t),
\end{align}
(\ref{eq_partialsquare}) becomes equivalent to 
\begin{align}
   \beta_0 =  \mathop{\arg}_{t} \left\{\frac{\partial\left(t e^{-\frac{\alpha^2+t^2}{2}}I_0(\alpha t)\right)}{\partial t}=0\right\}.
\end{align}
Using the approximation $I_0(t) \simeq \frac{e^t}{\sqrt{2\pi t}} $ \cite[Eq. (9.7.1)]{abramowitz1999ia} and writing
\begin{align}
    &~~~\frac{\partial\left(\sqrt{\frac{t}{2\pi \alpha}}e^{-\frac{(t-\alpha)^2}{2}}\right)}{\partial t} = 0\nonumber\\
   \Rightarrow &\frac{1}{\sqrt{2\pi\alpha}}\left(\frac{e^{-\frac{(t-\alpha)^2}{2}}}{2\sqrt{t}}+\sqrt{t}e^{-\frac{(t-\alpha)^2}{2}}(\alpha-t)\right)  = 0\nonumber\\
    \Rightarrow &2t^2-2\alpha t-1 =0,
\end{align}
we obtain 
\begin{align}\label{eq_beta0}
   \beta_0  = \frac{\alpha+\sqrt{\alpha^2+2}}{2}
\end{align}
since $\beta\geq0$. In this way, we find the point 
\begin{align}
    \mathcal{C}=\left(\beta_0, 1-Q_1\left(\alpha,\beta_0\right)\right).
\end{align}

To calculate the slope $m$ at the point $\mathcal{C}$, we plug (\ref{eq_beta0}) into (\ref{eq_derivativeMarcumQ}) leading to
\begin{align}\label{eq_m}
    m = \beta_0e^{-\frac{1}{2}\left(\alpha^2+\left(\beta_0\right)^2\right)}I_0\left(\alpha\beta_0\right).
\end{align}
Finally, using (\ref{eq_YY}), (\ref{eq_beta0}) and (\ref{eq_m}), the CDF $y(\alpha, \beta) = 1-Q_1(\alpha,\beta)$ can be approximated  as in (\ref{eq_lema1}). Note that, because the CDF is limited to the range [0 1], the boundaries $c_1$ and $c_2$ in (\ref{eq_lema1}) are obtained by finding $x$ when setting $y=0$ and $y=1$ in (\ref{eq_YY}), which leads to the semi-linear approximation as given in (\ref{eq_lema1}).
\end{proof}

To further simplify the calculation, considering different ranges of $\alpha$, we can simplify $\beta_0$ and, consequently, (\ref{eq_lema1}) as stated in the following.

\begin{corollary}\label{coro1}
For moderate/large values of $\alpha$, we have $y(\alpha,\beta)\simeq\tilde{\mathcal{Z}}(\alpha,\beta)$ where
\begin{align}\label{eq_coro1}
\tilde{\mathcal{Z}}(\alpha,\beta)&\simeq
\begin{cases}
0,  ~~~~~~~~~~~~~~~~~~~~~~~~~~~~\mathrm{if}~\beta < \tilde{c}_1  \\ 
\alpha e^{-\alpha^2}I_0(\alpha^2)(\beta-\alpha) + 
\frac{1}{2}\left(1-e^{-\alpha^2}I_0(\alpha^2)\right), \\ ~~~~~~~~~~~~~~~~~~~~~~~~~~~~~~~\mathrm{if}~\tilde{c}_1 \leq\beta\leq \tilde{c}_2 \\
1,  ~~~~~~~~~~~~~~~~~~~~~~~~~~~~~\mathrm{if}~ \beta> \tilde{c}_2.
\end{cases}\\
&\overset{(a)}{\simeq}
\begin{cases}
0, ~~~~~~~~~~~~~~~~~~~~~~~~~~~~~\mathrm{if}~  \beta < \breve{c}_1  \\ 
\frac{1}{\sqrt{2\pi}}(\beta-\alpha) + 
\frac{1}{2}\left(1-\frac{1}{\sqrt{2\pi\alpha^2}}\right), \\~~~~~~~~~~~~~~~~~~~~~~~~~~~~~~~\mathrm{if}~  \breve{c}_1 \leq\beta\leq \breve{c}_2 \\
1, ~~~~~~~~~~~~~~~~~~~~~~~~~~~~~\mathrm{if}~ \beta>\breve{c}_2,
\end{cases}
\end{align}
with $\tilde{c}_1$ and $\tilde{c}_2$ given in (\ref{eq_tildec1}) and (\ref{eq_tildec2}), $\breve{c}_1$ and $\breve{c}_2$ given in (\ref{eq_dotc1}) and (\ref{eq_dotc2}), respectively.
\end{corollary}

\begin{proof}
Using (\ref{eq_beta0}) for moderate/large values of $\alpha$, we have $\beta_0 \simeq \alpha$ and \cite[Eq.  (A-3-2)]{schwartz1995communication}
\begin{align}
    Q_1(\alpha,\alpha) = \frac{1}{2}\left(1+e^{-\alpha^2}I_0(\alpha^2)\right),
\end{align}
leading to
\begin{align}\label{eq_tildec1}
        \tilde{c}_1 = \frac{-\frac{1}{2}\left(1-e^{-\alpha^2}I_0(\alpha^2)\right)}{\alpha e^{-\alpha^2}I_0(\alpha^2)}+\alpha,
\end{align}
\begin{align}\label{eq_tildec2}
     \tilde{c}_2 = \frac{1-\frac{1}{2}\left(1-e^{-\alpha^2}I_0(\alpha^2)\right)}{\alpha e^{-\alpha^2}I_0(\alpha^2)}+\alpha,
\end{align}
and (\ref{eq_coro1}). Finally, $(a)$ is obtained by using the approximation $I_0(x) \simeq \frac{e^x}{\sqrt{2\pi x}}$ resulting in
\begin{align}\label{eq_dotc1}
        \breve{c}_1 = -\frac{\sqrt{2\pi}}{2}\left(1-\frac{1}{\sqrt{2\pi\alpha^2}}\right)+\alpha,
\end{align}
and
\begin{align}\label{eq_dotc2}
     \breve{c}_2 = \sqrt{2\pi}-\frac{\sqrt{2\pi}}{2}\left(1-\frac{1}{\sqrt{2\pi\alpha^2}}\right)+\alpha.
\end{align}
\end{proof}

\begin{corollary}\label{coro2}
For small values of $\alpha$, we have  $y(\alpha,\beta)\simeq\hat{ \mathcal{Z}}(\alpha,\beta)$ where
\begin{align}\label{eq_coro2}
\hat{ \mathcal{Z}}(\alpha,\beta)\simeq
\begin{cases}
0,   ~~~~~~~~~~~~~~~~~~~~~~~~~~~~~~~~~\mathrm{if}~\beta < \hat{c}_1  \\ 
\frac{\alpha+\sqrt{2}}{2}  e^{-\frac{\alpha^2+\left(\frac{\alpha+\sqrt{2}}{2}\right)^2}{2}}\times\\
~~~I_0\left(\frac{\alpha(\alpha+\sqrt{2})}{2}\right)(\beta-\frac{\alpha+\sqrt{2}}{2}) + \\
~~~1-Q_1\left(\alpha,\frac{\alpha+\sqrt{2}}{2}\right), 
~~~~~~~~\mathrm{if}~\hat{c}_1 \leq\beta\leq \hat{c}_2 \\
1, ~~~~~~~~~~~~~~~~~~~~~~~~~~~~~~~~~\mathrm{if}~ \beta> \hat{c}_2
\end{cases}
\end{align}
with $\hat{c}_1$ and $\hat{c}_2$ given in (\ref{eq_hatc1}) and (\ref{eq_hatc2}), respectively.
\end{corollary}
\begin{proof}
Using (\ref{eq_beta0}) for small values of $\alpha$, we have $\beta_0\simeq \frac{\alpha+\sqrt{2}}{2}$, which leads to 
\begin{equation}\label{eq_hatc1}
    \hat{c}_1 = \frac{-1+Q_1\left(\alpha,\frac{\alpha+\sqrt{2}}{2}\right)}{\left(\frac{\alpha+\sqrt{2}}{2} \right)^2 e^{-\frac{\alpha^2+\left(\frac{\alpha+\sqrt{2}}{2}\right)^2}{2}} I_0\left(\frac{\alpha(\alpha+\sqrt{2})}{2}\right)}+\alpha,
\end{equation}
and
\begin{equation}\label{eq_hatc2}
    \hat{c}_2 = \frac{Q_1\left(\alpha,\frac{\alpha+\sqrt{2}}{2}\right)}{\left(\frac{\alpha+\sqrt{2}}{2} \right)^2 e^{-\frac{\alpha^2+\left(\frac{\alpha+\sqrt{2}}{2}\right)^2}{2}} I_0\left(\frac{\alpha(\alpha+\sqrt{2})}{2}\right)}+\alpha,
\end{equation}
and simplifies (\ref{eq_lema1}) to (\ref{eq_coro2}).
\end{proof}


To illustrate these semi-linear approximations, Fig. \ref{fig_CDFilu} shows the CDF $y(\alpha,\beta)= 1-Q_1(\alpha,\beta)$ for both small and large values of $\alpha$, and compares the exact CDF with the approximation schemes of Lemma \ref{Lemma1} and Corollaries \ref{coro1}-\ref{coro2}.  Note that, the ranges of $\alpha$ and $\beta$ depend on applications, and in communication systems experiencing a Rician-type channel, $\alpha$ is typically small or moderate, e.g., less than 5 \cite{Bocus2013CLapproximation}, because it is a function of channel/channel gain. 

From Fig. \ref{fig_CDFilu}, we can observe that Lemma \ref{Lemma1} is tight for  different ranges of $\alpha$  and moderate values of $\beta$. Moreover, the tightness is improved as $\alpha$ decreases. Also,  Corollaries \ref{coro1}-\ref{coro2} provide good approximations for large and small values of $\alpha$, respectively. Then, the proposed approximations are not tight at the tails of the CDF.  However, as observed in \cite{Bocus2013CLapproximation,Fu2011GLOBECOMexponential,Makki2013TCfeedback,Makki2011Eurasipcapacity,Makki2018WCLwireless,Makki2016TVTperformance,Simon2003TWCsome,Suraweera2010TVTcapacity,Ma2000JSACunified,Digham2007TCenergy,Cao2016CLsolutions,sofotasios2015solutions,Azari2018TCultra,Alam2014INFOCOMWrobust,Gao2018IAadmm,Shen2018TVToutage,Song2017JLTimpact,Tang2019IAan}  and in the following, in different applications,  Marcum $Q$-function is normally combined with other functions which tend to zero at the tails of the CDF. In such cases, the inaccuracy of the approximation at the tails does not affect the tightness of the final result.

\begin{figure}
\centering
  \includegraphics[width=1.0\columnwidth]{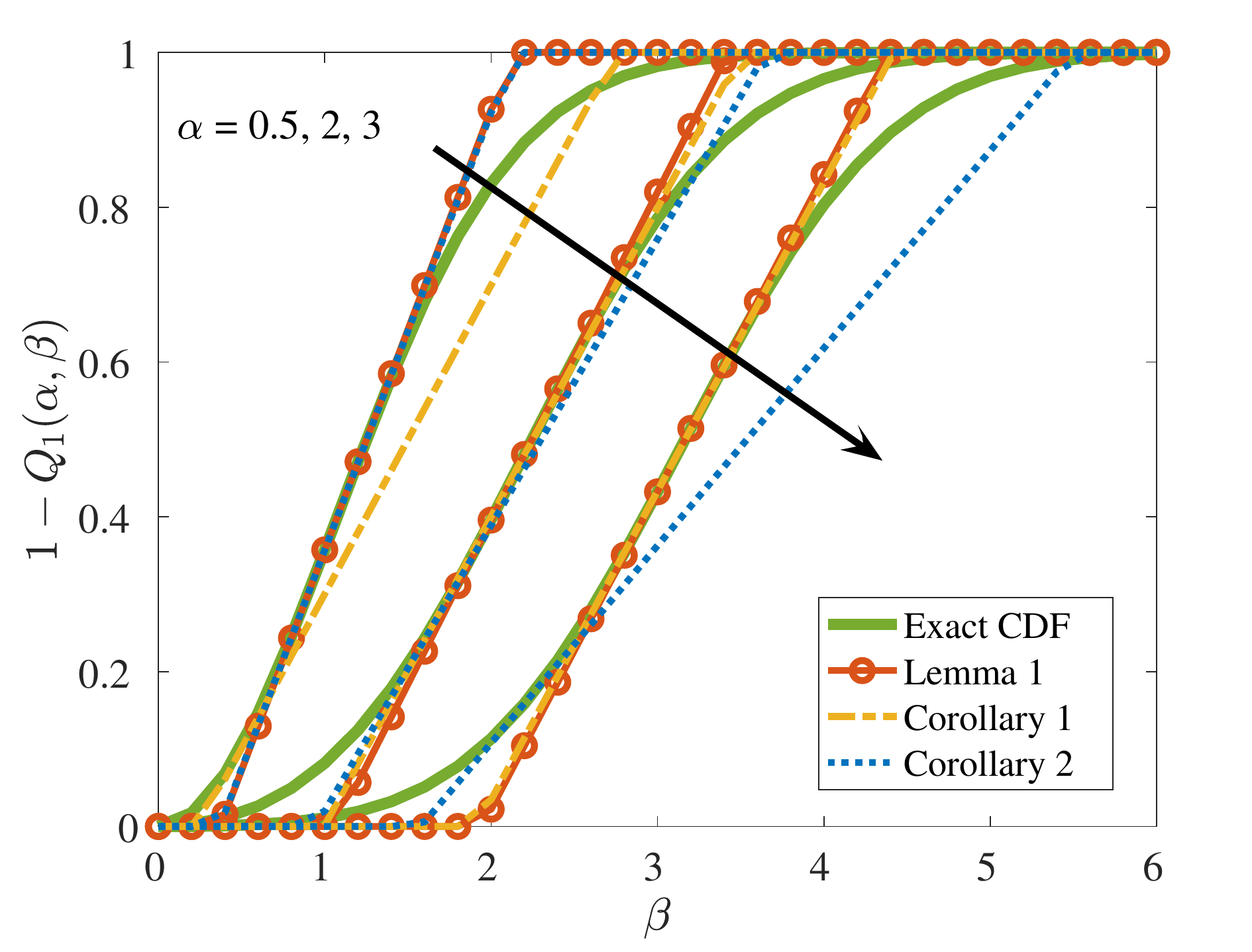}\\
\caption{Illustration of the semi-linear approximation with Lemma \ref{Lemma1}, and Corollaries \ref{coro1}-\ref{coro2}. For each value of $\alpha\in[0.5,2,3]$, the approximated results obtained by Lemma \ref{Lemma1} and Corollaries \ref{coro1}-\ref{coro2} are compared with the exact value for a broad range of $\beta$.}
\label{fig_CDFilu}
\end{figure}

As an example, we first consider a general integral in the form of
\begin{align}\label{eq_integral}
G(\alpha,\rho)=\int_\rho^\infty{e^{-nx} x^m \left(1-Q_1(\alpha,x)\right)\text{d}x} ~~\forall n,m,\alpha,\rho>0.
\end{align}
Such an integral has been observed in various applications, e.g., in bit-error-probability evaluation of a Rayleigh fading channel \cite[eq. (1) (13)]{Simon2003TWCsome}, in energy detection of unknown signals over various multipath fading channels \cite[eq. (2)]{Cao2016CLsolutions}, in capacity analysis with channel inversion and fixed rate over correlated Nakagami fading \cite[eq. (1)]{sofotasios2015solutions}, in performance evaluation of  incoherent receivers in radar systems \cite[eq. (3)]{Cui2012ELtwo}, and in error probability analysis of diversity receivers \cite[eq. (1)]{Gaur2003TVTsome}. However, depending on the values of $n, m$ and $\rho$, (\ref{eq_integral}) may have no closed-form expression. Using Corollary \ref{coro1}, $G(\alpha,\rho)$ can be approximated in closed-form as presented in Lemma \ref{Lemma2}.

\begin{lem}\label{Lemma2}
The integral (\ref{eq_integral}) is approximately given by
\begin{align}
G(\alpha,\rho)\simeq
\begin{cases}
\Gamma(m+1,n\rho)n^{-m-1}, ~~~~~~~~~~~~~~~~~~\mathrm{if}~  \rho \geq \breve{c}_2  \\ 
\Gamma(m+1,n\breve{c}_2)n^{-m-1} + \\
~~\left(-\frac{\alpha}{\sqrt{2\pi}}+\frac{1}{2}\left(1-\frac{1}{\sqrt{2\pi\alpha^2}}\right)\right)\times n^{-m-1}\times\\
~~\left(\Gamma(m+1,n\max(\breve{c}_1,\rho))-\Gamma(m+1,n\breve{c}_2)\right)+\\
~~\left(\Gamma(m+2,n\max(\breve{c}_1,\rho))-\Gamma(m+2,n\breve{c}_2)\right)\times\\
~~\frac{n^{-m-2}}{\sqrt{2\pi}},
~~~~~~~~~~~~~~~~~~~~~~~~~~~~~~~~\mathrm{if}~  \rho<\breve{c}_2,
\end{cases}
\end{align}
where $\Gamma(s,x) = \int_{x}^{\infty} t^{s-1}e^{-t} \mathrm{d}t$ is the upper incomplete gamma function \cite[Eq. 6.5.1]{abramowitz1999ia}.
\end{lem}
\begin{proof}
See Appendix \ref{proof_Lemma2}. 
\end{proof}

As a second integration example of interest, consider
\begin{align}\label{eq_integralT}
    T(\alpha,m,a,\theta_1,\theta_2) = \int_{\theta_1}^{\theta_2} e^{-mx}\log(1+ax)Q_1(\alpha,x)\text{d}x \nonumber\\\forall m>0,a,\alpha,
\end{align}
with $\theta_2>\theta_1\geq0$, which does not have a closed-form expression for different values of $m, a, \alpha$. This type of integral is interesting as it could be used to analyze the expected performance of outage-limited systems, e.g,  the considered integrals in the shape of \cite[eq. (1) (13)]{Simon2003TWCsome}, \cite[eq. (2)]{Cao2016CLsolutions}, \cite[eq. (3)]{Cui2012ELtwo}, and \cite[eq. (1)]{Gaur2003TVTsome}, applied in the analysis of the outage-limited throughput, i.e., when the outage-limited throughput $\log(1+ax)Q_1(\alpha,x)$   is averaged over fading statistics \cite[p. 2631]{Biglieri1998TITfading}\cite[Theorem 6]{Verdu1994TITgeneral}\cite[Eq. (9)]{Makki2014TCperformance}. Then, using Lemma \ref{Lemma1}, (\ref{eq_integralT}) can be approximated in closed-form as follows.

\begin{lem}\label{Lemma3}
The integral (\ref{eq_integralT}) is approximately given by
\begin{align}
  T(\alpha,m,a,\theta_1,\theta_2)\simeq~~~~~~~~~~~~~~~~~~~~~~~~~~~~~~~~~~~~~~~~~~~~~~~~~~\nonumber\\
\begin{cases}
\mathcal{F}_1(\theta_2)-\mathcal{F}_1(\theta_1), ~~~~~~~~~~~~~~~~~~~~~~~~~\mathrm{if}~  0\leq\theta_1<\theta_2 < c_1  \\ 
\mathcal{F}_1(c_1)-\mathcal{F}_1(\theta_1)+\mathcal{F}_2(\max(c_2,\theta_2))-\mathcal{F}_2(c_1),  \\~~~~~~~~~~~~~~~~~~~~~~~~~~~~~~~~~~~~~~~~~~~~~~\mathrm{if}~ \theta_1<c_1, \theta_2\geq c_1\\
 \mathcal{F}_2(\max(c_2,\theta_2))-\mathcal{F}_2(c_1), \\~~~~~~~~~~~~~~~~~~~~~~~~~~~~~~~~~~~~~~~~~~~~~~\mathrm{if}~ \theta_1>c_1\\
0, ~~~~~~~~~~~~~~~~~~~~~~~~~~~~~~~~~~~~~~~~~~~~\mathrm{if}~ \theta_1 > c_2,
\end{cases}  
\end{align}
where 
$c_1$ and $c_2$ are given by (\ref{eq_c1}) and (\ref{eq_c2}), respectively. Moreover,
\begin{align}\label{eq_F1}
    \mathcal{F}_1(x) \doteq \frac{1}{m}\left(-e^{\frac{m}{a}}\operatorname{E_1}\left(mx+\frac{m}{a}\right)-e^{-mx}\log(ax+1)\right),
\end{align}
and
\begin{align}\label{eq_F2}
    \mathcal{F}_2(x) \doteq &~ \mathrm{e}^{-mx}\Bigg(\left(mn_2-an_2-amn_1\right)\mathrm{e}^\frac{m\left(ax+1\right)}{a}\nonumber\\&~~ \operatorname{E_1}\left(\frac{m\left(ax+1\right)}{a}\right)-
    a\left(mn_2x+n_2+mn_1\right)\nonumber\\&~~\log\left(ax+1\right)-an_2\Bigg),
\end{align}
with
\begin{align}\label{eq_n1}
    n_1 = 1+\beta_0 e^{-\frac{1}{2}\left(\alpha^2+\left(\beta_0\right)^2\right)}
 I_0\left(\alpha\beta_0\right)\beta_0-
 1+Q_1\left(\alpha,\beta_0\right),
\end{align}
and
\begin{align}\label{eq_n2}
    n_2 = -\beta_0 e^{-\frac{1}{2}\left(\alpha^2+\left(\beta_0\right)^2\right)}I_0\left(\alpha\beta_0\right).
\end{align}
In (\ref{eq_F1}) and (\ref{eq_F2}), $\operatorname{E_1}(x) = \int_x^{\infty} \frac{e^{-t}}{t} \mathrm{d}t$ is the Exponential Integral function \cite[p. 228, (5.1.1)]{abramowitz1999ia}.
\end{lem}

\begin{proof}
See Appendix \ref{proof_Lemma3}.
\end{proof}

Finally, setting $m = 0$ in (\ref{eq_integralT}), i.e.,
\begin{align}\label{eq_integralTs}
     T(\alpha,0,a,\theta_1,\theta_2) = \int_{\theta_1}^{\theta_2} \log(1+ax)Q_1(\alpha,x)\text{d}x,  \forall a,\alpha,
\end{align}
one can follow the same procedure in (\ref{eq_integralT}) to approximate (\ref{eq_integralTs}) as
\begin{align}\label{eq_integralTss}
  T(\alpha,0,a,\theta_1,\theta_2)\simeq~~~~~~~~~~~~~~~~~~~~~~~~~~~~~~~~~~~~~~~~~~~~~~~~~~\nonumber\\
\begin{cases}
\mathcal{F}_3(\theta_2)-\mathcal{F}_3(\theta_1), ~~~~~~~~~~~~~~~~~~~~~~~~~\mathrm{if}~  0\leq\theta_1<\theta_2 < c_1  \\ 
\mathcal{F}_3(c_1)-\mathcal{F}_3(\theta_1)+\mathcal{F}_4(\max(c_2,\theta_2))-\mathcal{F}_4(c_1),  \\~~~~~~~~~~~~~~~~~~~~~~~~~~~~~~~~~~~~~~~~~~~~~~\mathrm{if}~ \theta_1<c_1, \theta_2\geq c_1\\
 \mathcal{F}_4(\max(c_2,\theta_2))-\mathcal{F}_4(c_1), \\~~~~~~~~~~~~~~~~~~~~~~~~~~~~~~~~~~~~~~~~~~~~~~\mathrm{if}~ \theta_1>c_1\\
0, ~~~~~~~~~~~~~~~~~~~~~~~~~~~~~~~~~~~~~~~~~~~\mathrm{if}~ \theta_1 > c_2,
\end{cases}  
\end{align}
with $c_1$ and $c_2$ given by (\ref{eq_c1}) and (\ref{eq_c2}), respectively. Also,
\begin{align}
    \mathcal{F}_3 = \frac{(ax+1)(\log(ax+1)-1)}{a},
\end{align}
and
\begin{align}
    \mathcal{F}_4 = \frac{n_2\left((2a^2x^2-2)\log(ax+1)-a^2x^2+2ax\right)}{4a^2}+\nonumber\\\frac{n_1(ax+1)(\log(ax+1)-1)}{a},
\end{align}
where $n_1$ and $n_2$ are given by (\ref{eq_n1}) and (\ref{eq_n2}), respectively.

\begin{figure}
\centering
  \includegraphics[width=1.0\columnwidth]{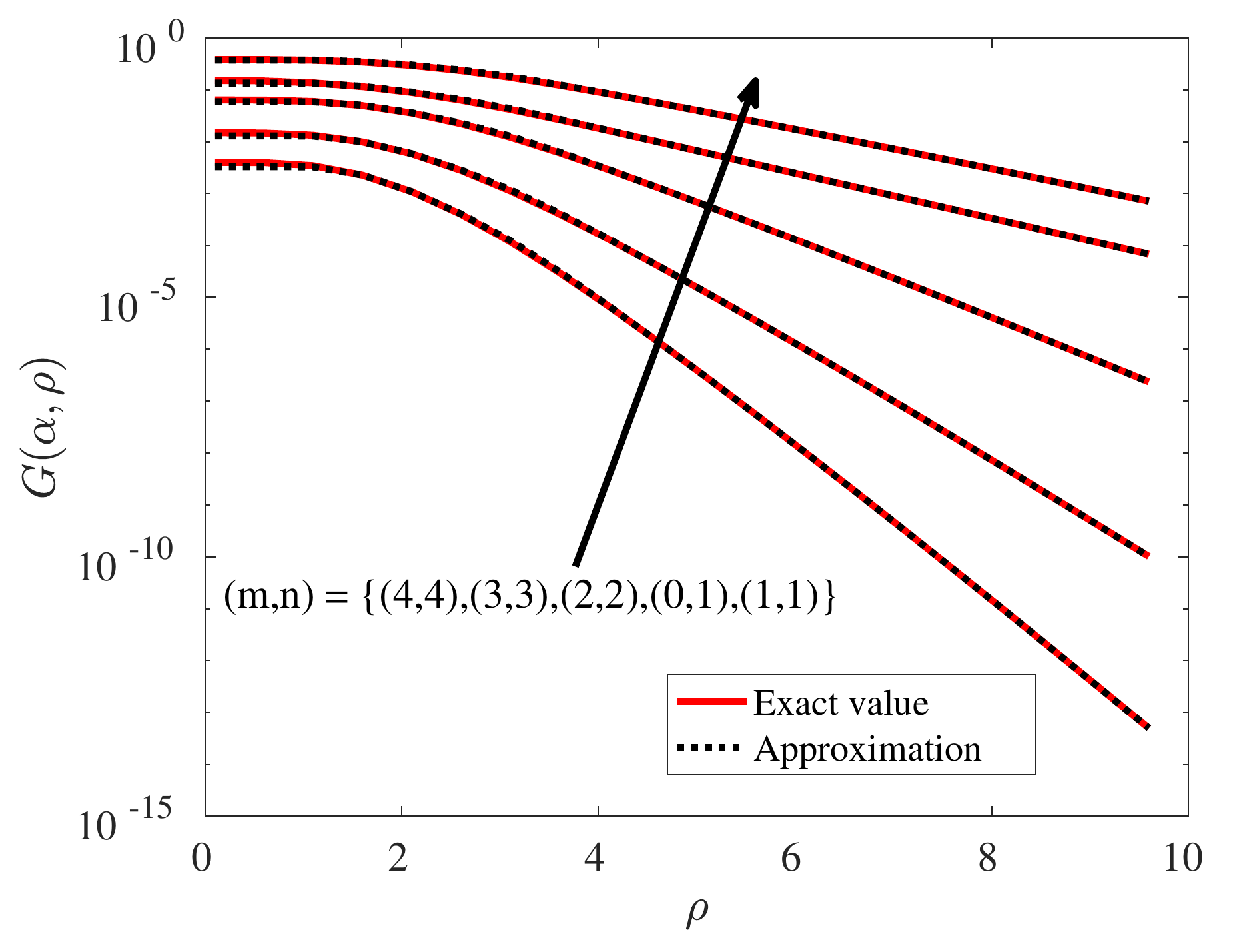}\\
\caption{The  integral (\ref{eq_integral}) involving Marcum $Q$-function. Solid lines are exact values while crosses are the results obtained from Lemma \ref{Lemma2}, $\alpha = 2$,  $(m,n) = \{(4,4), (3,3), (2,2), (0,1), (1,1)\}$.}\label{fig_integrald}
\end{figure}

\begin{figure}
\centering
  \includegraphics[width=1.0\columnwidth]{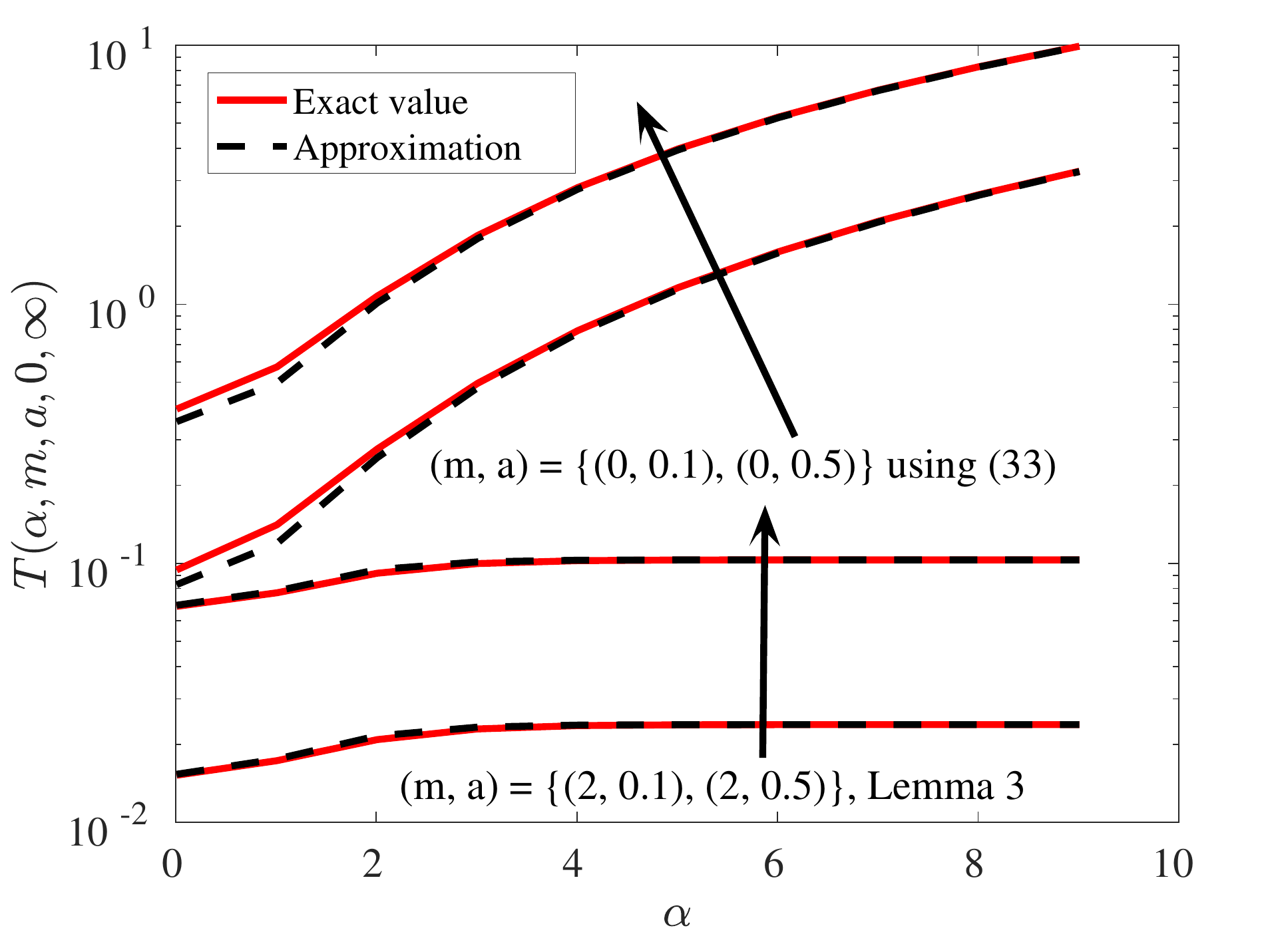}\\
\caption{The integral (\ref{eq_integralT}) involving the Marcum $Q$-function. Solid lines are exact values while crosses are the results obtained from Lemma \ref{Lemma3} and (\ref{eq_integralTss}). $\theta_1 = 0, \theta_2 = \infty$. }\label{fig_t2}
\end{figure}

In Figs. \ref{fig_integrald} and \ref{fig_t2}, we evaluate the tightness of the approximations in Lemmas \ref{Lemma2}, \ref{Lemma3} and (\ref{eq_integralTss}), for different values of $m$, $n$, $\rho$, $a$ and $\alpha$. From the figures, it can be observed that the approximation schemes of   Lemmas \ref{Lemma2}-\ref{Lemma3} and (\ref{eq_integralTss}) are very tight for different parameter settings, while our proposed semi-linear approximation makes it possible to represent the integrals in closed-form.  In this way, although the approximation (\ref{eq_lema1}) is not tight at the tails of the CDF, it gives tight approximation results when it appears in different integrals (Lemmas \ref{Lemma2}-\ref{Lemma3}) with the Marcum-$Q$ function  combined with other functions tending to zero at the tails of the function. Also, as we show in Section \ref{sec:application}, the semi-linear approximation scheme is efficient in optimization problems consisting of  Marcum $Q$-function. Finally, to tightly approximate the Marcum $Q$-function at the tails, which is the range of interest in, e.g., error probability analysis, one can use the approximation schemes of \cite{Simon2000TCexponential,annamalai2001WCMCcauchy}.

\section{Applications in PA Systems}\label{sec:application}
In Section \ref{sec:appro}, we showed how the proposed approximation scheme enables us to derive closed-form expressions for a broad range of integrals, as required in various expectation-based calculations, e.g., \cite{Simon2003TWCsome,Cao2016CLsolutions,sofotasios2015solutions,Cui2012ELtwo,Gaur2003TVTsome,Simon2000TCexponential,6911973}. On the other hand,  Marcum $Q$-function may also  appear in  optimization problems, e.g., \cite[eq. (8)]{Azari2018TCultra},\cite[eq. (9)]{Alam2014INFOCOMWrobust}, \cite[eq. (10)]{Gao2018IAadmm}, \cite[eq. (10)]{Shen2018TVToutage}, \cite[eq. (15)]{Song2017JLTimpact}, \cite[eq. (22)]{Tang2019IAan}.  For this reason, in this section, we provide an example of using our proposed semi-linear approximation in an optimization problem for PA systems.

\subsection{Problem Formulation}
Vehicle communication is one of the most important use cases in 5G.  Here, the main focus is to provide efficient and reliable connections to cars and public transports, e.g., busses and trains. CSIT plays an important role in achieving these goals, since the data transmission efficiency can be improved by updating the transmission parameters relative to the instantaneous channel state.  However, the typical CSIT acquisition systems, which are mostly designed for (semi)static channels, may not work well for high-speed vehicles. This is because, depending on the vehicle speed, the position of the antennas may change quickly and the channel information becomes inaccurate.  To overcome this issue, \cite{Sternad2012WCNCWusing,DT2015ITSMmaking,BJ2017PIMRCpredictor,phan2018WSAadaptive,Jamaly2014EuCAPanalysis, BJ2017ICCWusing} propose the PA setup as shown in Fig. \ref{system}. With a PA setup, which is of interest in cellular vehicle-to-everything (C-V2X) communications \cite{Sternad2012WCNCWusing} as well as integrated access and backhauling \cite{Teyeb2019VTCintegrated,madaapatha2020iab}, two antennas are deployed on the top of the vehicle. The first antenna,  the PA,  estimates the channel and sends feedback to the BS at time $t$. Then, the BS uses the CSIT provided by the PA to communicate with a second antenna, which we refer to as RA, at time $t+\delta$, where $\delta$ is the processing time at the BS. In this way, the BS can use the CSIT acquired from the PA and perform various CSIT-based transmission schemes, e.g.,  \cite{Sternad2012WCNCWusing,BJ2017ICCWusing}.

\begin{figure}
\centering
  \includegraphics[width=1.0\columnwidth]{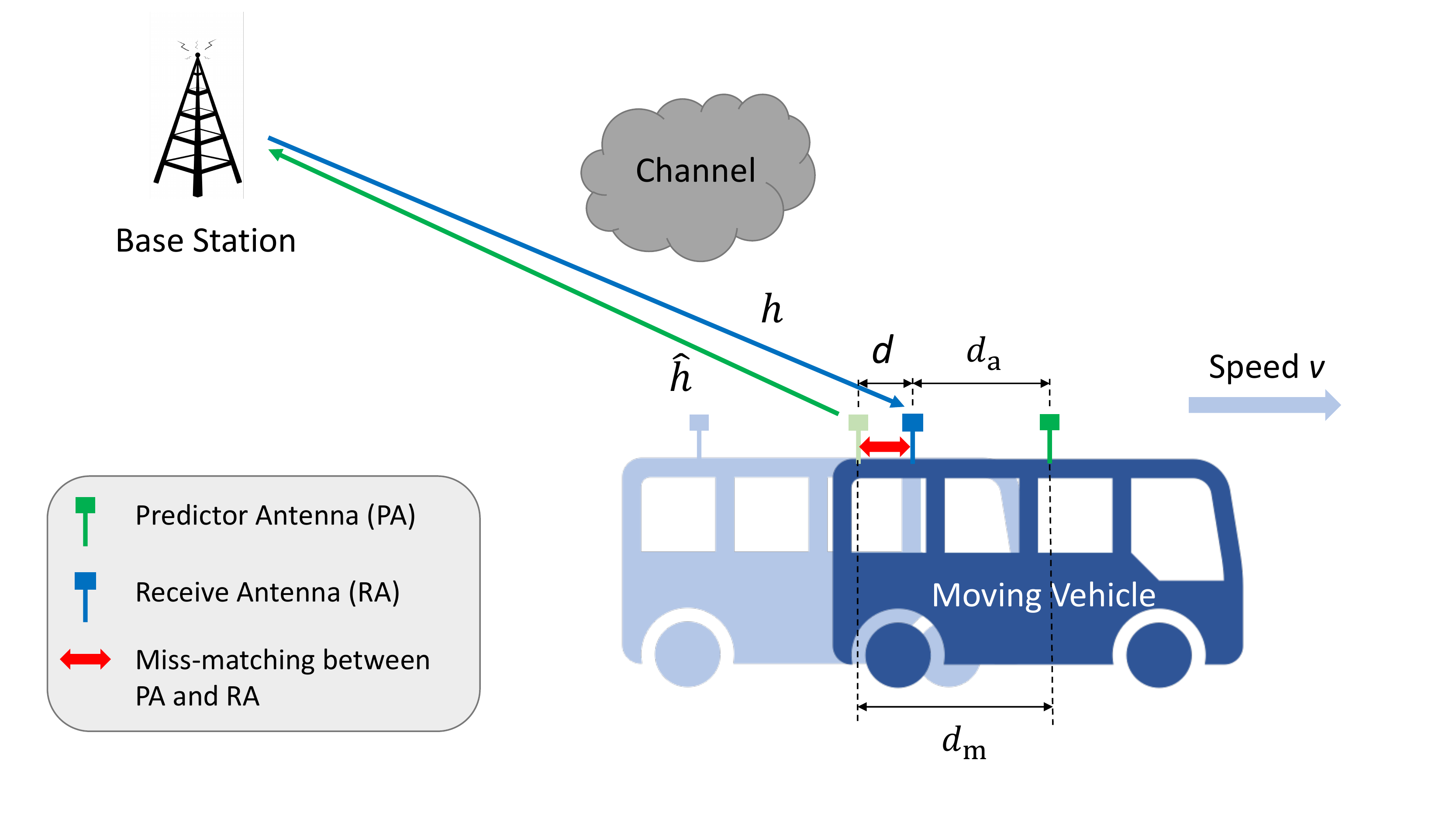}\\
\caption{PA system with the mismatch problem. Here, $\hat h$ is the channel between the BS and the PA while $h$ refers to the channel in the BS-RA link. The vehicle is moving with speed $v$ and the antenna separation is $d_\text{a}$. The red arrow indicates the spatial mismatch, i.e., when the RA does not reach  the same  point as the PA when sending pilots. Also, $d_\text{m}$ is the moving distance of the vehicle which is affected by the processing delay $\delta$ of the BS. }\label{system}
\end{figure}

We assume that the vehicle moves through a stationary electromagnetic standing wave pattern, as this assumption has been experimentally validated to be essentially correct in \cite{Jamaly2014EuCAP}. Thus, if the RA reaches exactly the same position as the position of the PA when sending the pilots, it will experience the same channel and the CSIT will be perfect\footnote{In our previous work \cite[Sec. III.A]{Guo2019WCLrate} we have studied the effect of non-stationary environment on the system performance.}. However, if the RA does not reach the same spatial point as the PA, due to, e.g., the BS processing delay is not equal to the time that we need until the RA reaches the same point as the PA, the RA may receive the data in a place different from the one in which the PA was sending the pilots. Such spatial mismatch may lead to CSIT inaccuracy, which will affect the system performance considerably. Thus, we need a method to compensate for it, and here we analyze optimal rate adaptation.

Considering downlink transmission in the BS-RA link, the received signal is given by
\begin{align}\label{eq_Y}
{{y}} = \sqrt{P}hx + w.
\end{align}
Here, $P$ represents the transmit power, $x$ is the input message with unit variance, and $h$ is the fading coefficient between the BS and the RA. Also, $w \sim \mathcal{CN}(0,1)$ denotes the independent and identically distributed (IID) complex Gaussian noise added at the receiver.

We denote the channel coefficient of the PA-BS uplink as $\hat{h}$. Also, we define $d$ as the effective distance between  the place where the PA estimates the channel at time $t$, and the place where the RA reaches at time $t+\delta$. As can be seen in Fig. \ref{system}, $d$ can be calculated as
\begin{align}\label{eq_d}
    d = |d_\text{a} - d_\text{m} | = |d_\text{a} - v\delta|,
\end{align}
where $d_\text{m}$ is the moving distance of the vehicle during time interval $\delta$, and $v$ is the velocity of the vehicle. Also, $d_\text{a}$ is the antenna separation between the PA and the RA. In conjunction to (\ref{eq_d}), here, we assume $d$ can be calculated by the BS. 

Using the classical Jake's correlation model \cite[p. 2642]{Shin2003TITcapacity} and  assuming a semi-static propagation environment, i.e., assuming that the coherence time of the propagation environment is larger than $\delta$,  the channel coefficient of the BS-RA downlink can be modeled as 
\begin{align}\label{eq_H}
    h = \sqrt{1-\sigma^2} \hat{h} + \sigma q.
\end{align}
Here, $q \sim \mathcal{CN}(0,1)$  is independent of the known channel value $\hat{h}\sim \mathcal{CN}(0,1)$, and $\sigma$ is a function of the effective distance $d$ as \cite[p. 2642]{Shin2003TITcapacity}
\begin{align}
    \sigma = \frac{\frac{\phi_2^2-\phi_1^2}{\phi_1}}{\sqrt{ \left(\frac{\phi_2}{\phi_1}\right)^2 + \left(\frac{\phi_2^2-\phi_1^2}{\phi_1}\right)^2 }} = \frac{\phi_2^2-\phi_1^2}{\sqrt{ \left(\phi_2\right)^2 + \left(\phi_2^2-\phi_1^2\right)^2 }} .
\end{align}
Here, $\phi_1 = \bm{\Phi}_{1,1}^{1/2} $ and $\phi_2 = \bm{\Phi}_{1,2}^{1/2} $, where $\bm{\Phi}$ is  from Jake's model \cite[p. 2642]{Shin2003TITcapacity}
\begin{align}\label{eq_tildeH}
     \bigl[ \begin{smallmatrix}
  \hat{h}\\h
\end{smallmatrix} \bigr]= \bm{\Phi}^{1/2} \bm{h}_{\varepsilon}.
\end{align}
Note that, the spatial mismatch phenomenon (\ref{eq_H}) has been experimentally verified in, e.g., \cite{Jamaly2014EuCAP} for PA setups. Also, one can follow the same method as in \cite{Guo2019WCLrate} to extend the model to the cases with temporally-correlated channels. Moreover, in (\ref{eq_tildeH}), $\bm{h}_{\varepsilon}$ has independent circularly-symmetric zero-mean complex Gaussian entries with unit variance, and $\bm{\Phi}$ is the channel correlation matrix with the $(i,j)$-th entry given by
\begin{align}\label{eq_phi}
    \Phi_{i,j} = J_0\left((i-j)\cdot2\pi d/ \lambda\right) \forall i,j.
\end{align}
Here, $J_n(x) = (\frac{x}{2})^n \sum_{i=0}^{\infty}\frac{(\frac{x}{2})^{2i}(-1)^{i} }{i!\Gamma(n+i+1)}$ represents the $n$-th order Bessel function of the first kind. Moreover, $\lambda$ denotes the carrier wavelength, i.e., $\lambda = c/f_\text{c}$ where $c$ is the speed of light and $f_\text{c}$ is the carrier frequency.

From (\ref{eq_H}), for a given $\hat{h}$ and $\sigma \neq 0$, $|h|$ follows a Rician distribution, i.e., the  probability density function (PDF) of $|h|$ is given by 
\begin{align}
    f_{|h|\big|\hat{g}}(x) = \frac{2x}{\sigma^2}e^{-\frac{x^2+\hat{g}}{\sigma^2}}I_0\left(\frac{2x\sqrt{\hat{g}}}{\sigma^2}\right),
\end{align}
where $\hat{g} = (1-\sigma^2)|\hat{h}|^2$. Let us define the channel gain between BS-RA as $ g = |{h}|^2$. Then, the PDF of $g|\hat{g}$ is given by

\begin{align}\label{eq_pdf}
    f_{g|\hat{g}}(x) = \frac{1}{\sigma^2}e^{-\frac{x+\hat{g}}{\sigma^2}}I_0\left(\frac{2\sqrt{x\hat{g}}}{\sigma^2}\right),
\end{align}
which is non-central Chi-squared distributed with the CDF containing the first-order Marcum $Q$-function as
\begin{align}\label{eq_cdf}
    F_{g|\hat{g}}(x) = 1 - Q_1\left( \sqrt{\frac{2\hat{g}}{\sigma^2}}, \sqrt{\frac{2x}{\sigma^2}}  \right).
\end{align}

\subsection{Analytical Results on Rate Adaptation Using the Semi-Linear Approximation of the First-order Marcum Q-Function}\label{In Section III.C,}
We assume that $d_\text{a}$, $\delta $ and $\hat{g}$ are known by the BS. It can be seen from (\ref{eq_pdf}) that $f_{g|\hat{g}}(x)$ is a function of $v$. For a given $v$, the distribution of $g$ is known by the BS, and a rate adaption scheme can be performed to improve the system performance.

For a given instantaneous value of $\hat g$, the data is transmitted with instantaneous rate $R_{|\hat{g}}$ nats-per-channel-use (npcu). If the instantaneous channel gain realization supports the transmitted data rate $R_{|\hat{g}}$, i.e., $\log(1+gP)\ge R_{|\hat{g}}$, the data can be successfully decoded. Otherwise,  outage occurs. Hence, the outage probability in each time slot is
\begin{align}
    \Pr(\text{outage}|\hat{g}) = F_{g|\hat{g}}\left(\frac{e^{R_{|\hat{g}}}-1}{P}\right).
\end{align}
Also, the instantaneous throughput for a given $\hat{g}$ is
\begin{align}\label{eq_opteta}
\eta_{|\hat {g}}\left(R_{|\hat{g}}\right)=R_{|\hat{g}}\left(1-\Pr\left(\log(1+gP)<R_{|\hat{g}}\right)\right),
\end{align}
and the optimal rate adaptation maximizing the  instantaneous throughput is obtained by
\begin{align}\label{eq_optR}
    R_{|\hat{g}}^{\text{opt}}&=\argmax_{R_{|\hat{g}}\geq 0} \left\{ \left(1-\Pr\left(\log(1+gP)<R_{|\hat{g}}\right)\right)R_{|\hat{g}} \right\}\nonumber\\
    &=\argmax_{R_{|\hat{g}}\geq 0} \left\{ \left(1-F_{g|\hat{g}}\left(\frac{e^{R_{|\hat{g}}}-1}{P}\right)\right) R_{|\hat{g}}\right\}\nonumber\\
    &=\argmax_{R_{|\hat{g}}\geq 0} \left\{ Q_1\left(  \sqrt{\frac{2\hat{g}}{\sigma^2}}, \sqrt{\frac{2(e^{R_{|\hat{g}}}-1)}{P\sigma^2}} \right)R_{|\hat{g}} \right\},
\end{align}
where the last equality comes from (\ref{eq_cdf}). 

Using the derivatives of the Marcum $Q$-function, (\ref{eq_optR}) does not have a closed-form solution\footnote{To solve (\ref{eq_optR}), one can use different approximation methods of the Marcum $Q$-function, e.g., \cite{Guo2019WCLrate}. Here, we use Lemma \ref{Lemma1}/Corollaries \ref{coro1}-\ref{coro2} to solve (\ref{eq_optR}) to show the usefulness of the semi-linear approximation method. As shown, in Figs. \ref{fig_Figure5} and \ref{fig_Figure6}, the semi-linear approximation is tight for a broad range of parameter settings while it simplifies the optimization problem considerably. }. For this reason,  Lemma \ref{Lemma4}  uses the semi-linear approximation scheme of Lemma \ref{Lemma1} and Corollaries \ref{coro1}-\ref{coro2} to find the optimal data rate maximizing the instantaneous throughput.

\begin{lem}\label{Lemma4}
For a given channel realization $\hat{g}$, the throughput-optimized rate allocation is approximately given by
\begin{align}\label{eq_appRF}
     R_{|\hat{g}}^{\text{opt}}
     \simeq 2\mathcal{W}\left(\frac{(1+o_1o_2-o_3)e\sqrt{2P\sigma^2}}{2o_1}-1\right),
\end{align}
where $\mathcal{W}(\cdot)$ denotes the Lambert $\mathcal{W}$-function.
\end{lem}
\begin{proof}
The  approximation results of Lemma \ref{Lemma1} and Corollaries \ref{coro1}-\ref{coro2} can be generalized by $y(\alpha,\beta)\simeq\tilde{\mathcal{Z}}_{\text{general}}(\alpha,\beta)$ where
\begin{align}
\tilde{\mathcal{Z}}_{\text{general}}(\alpha,\beta)\simeq
\begin{cases}
0,~~~~~~~~~~~~~~~~~~~~\mathrm{if}~\beta < c_1(\alpha)  \\ 
o_1(\alpha)(\beta-o_2(\alpha)) +o_3,~~~\mathrm{else}\\ 
1,~~~~~~~~~~~~~~~~~~~~\mathrm{if}~ \beta> c_2(\alpha).
\end{cases}
\end{align}
$o_i,i=1,2,3$, are given by (\ref{eq_lema1}), (\ref{eq_coro1}), or (\ref{eq_coro2}) depending on if we use Lemma \ref{Lemma1} or Corollaries \ref{coro1}-\ref{coro2}. In this way, (\ref{eq_opteta}) is approximated as
\begin{align}\label{eq_appR}
    \eta_{|\hat {g}}\simeq R_{|\hat{g}}\left(1-o_1(\alpha)\beta + o_1(\alpha)o_2(\alpha) - o_3(\alpha)\right),
\end{align}
where $\alpha = \sqrt{\frac{2\hat{g}}{\sigma^2}}$. To simplify the equation, we omit $\alpha$ in the following since it is a constant for given $\hat{g}$, $\sigma$. Then, setting the derivative of (\ref{eq_appR}) with respect to $R$ equal to zero, we obtain
\begin{align}
    & R_{|\hat{g}}^{\text{opt}}  \nonumber\\
    & = \operatorname*{arg}_{R_{|\hat{g}}\geq 0}\left\{ 1+o_1o_2-o_3-o_1\left(\frac{(R_{|\hat{g}}+2)e^{R_{|\hat{g}}}-2}{\sqrt{2P\sigma^2\left(e^{R_{|\hat{g}}}-1\right)}}\right)=0\right\}\nonumber\\
    & \overset{(b)}{\simeq} \operatorname*{arg}_{R_{|\hat{g}}\geq 0}\left\{ \left(\frac{R_{|\hat{g}}}{2}+1\right)e^{\frac{R_{|\hat{g}}}{2}+1} = \frac{(1+o_1o_2-o_3)e\sqrt{2P\sigma^2}}{2o_1}\right\}\nonumber\\
    & \overset{(c)}{=} 2\mathcal{W}\left(\frac{(1+o_1o_2-o_3)e\sqrt{2P\sigma^2}}{2o_1}-1\right).
\end{align}
Here,  $(b)$ comes from $e^{R_{|\hat{g}}}-1 \simeq e^{R_{|\hat{g}}} $ and $(R_{|\hat{g}}+2)e^{R_{|\hat{g}}}-2 \simeq (R_{|\hat{g}}+2)e^{R_{|\hat{g}}} $ which are appropriate at moderate/high values of $R_{|\hat{g}}$. Also, $(c)$ is obtained by the definition of the Lambert $\mathcal{W}$-function $xe^x = y \Leftrightarrow x = \mathcal{W}(y)$ \cite{corless1996lambertw}. 

\end{proof}

Finally, the expected throughput, averaged over multiple time slots, is obtained by 
\begin{align}\label{eq_etaexpected}
  \eta = \mathbb{E}\left\{\eta|{\hat {g}}(R_{|\hat{g}}^{\text{opt}})\right\}  
\end{align}
with expectation over $\hat g$.  

Using (\ref{eq_appRF}) and the approximation \cite[Theorem. 2.1]{hoorfar2007approximation}
\begin{align}
    \mathcal{W}(x) \simeq \log(x)-\log\log(x), x\geq0,
\end{align}
we obtain
\begin{align}
    R_{|\hat{g}}^{\text{opt}}\simeq 2\log\left(\frac{(1+o_1o_2-o_3)e\sqrt{2P\sigma^2}}{2o_1}-1\right)-\nonumber\\
    2\log\log\left(\frac{(1+o_1o_2-o_3)e\sqrt{2P\sigma^2}}{2o_1}-1\right)
\end{align}
which implies that as the transmit power increases, the optimal instantaneous rate increases with the square root of the transmit power (approximately) logarithmically.

\subsection{On the Effect of Imperfect Channel Estimation}
In Section \ref{In Section III.C,},  for simplicity of discussions, we assumed perfect channel estimation of the BS-PA channel, i.e., $\hat{h}$, at the BS. The experiments in e.g., \cite{Jamaly2014EuCAP}, show that
we may not achieve perfect correlation although we measure at the correct location.  For example, in the case of FDD, the error sources for the CSIT $\hat{h}$ at the BS of the BS-PA channel could be channel estimation error at the PA. These deviations of channel estimation would invalidate the assumption of perfect channel information, and should be considered in the system design. Here, we follow the similar approach as in, e.g., \cite{Wang2007TWCperformance}, to add the effect of estimation error of $\hat{h}$ as an independent additive Gaussian variable whose variance is given by the accuracy of channel estimation. 

Let us define $\tilde h$ as the estimate of $\hat h$ at the BS. Then, following \cite{Wang2007TWCperformance}, we further develop our channel model  (\ref{eq_H}) as
\begin{align}\label{eq_Htp}
    \tilde{h} = \kappa \hat{h} + \sqrt{1-\kappa^2} z, 
\end{align}
for each time slot, where $z \sim \mathcal{CN}(0,1)$ is a Gaussian noise which is uncorrelated with $\hat{h}$. Also, $\kappa$ is a known correlation factor which represents the estimation error of $\hat{h}$ by $\kappa = \frac{\mathbb{E}\{\tilde{h}\hat{h}^*\}}{\mathbb{E}\{|\hat{h}|^2\}}$, i.e., the estimation error decreases with $\kappa$. Substituting (\ref{eq_Htp}) into (\ref{eq_H}) to replace $\hat{h}$ as $\kappa \hat{h} + \sqrt{1-\kappa^2} z$, we have the following estimate of the BS-RA channel at the BS side
\begin{align}\label{eq_Ht}
    h = \kappa\sqrt{1-\sigma^2}\hat{h}+\kappa\sigma q+\sqrt{1-\kappa^2}z.
\end{align}
Then, because  $\kappa\sigma q + \sqrt{1-\kappa^2}z$ is equivalent to a new Gaussian variable $w \sim\mathcal{CN}\left(0,(\kappa\sigma)^2+1-\kappa^2\right)$, we can follow the same procedure as in (\ref{eq_optR})-(\ref{eq_appRF}) to analyze the system performance with imperfect channel estimation of the PA (see Figs. \ref{fig_Figure5}-\ref{fig_Figure6} for more discussions). Finally, note that, one can involve additional independent Gaussian variables to model other error sources, e.g., imperfect antenna coupling \cite{Jamaly2019IETeffects}, and follow the same procedure as in (\ref{eq_Ht}) to evaluate the performance.

\subsection{Simulation Results}
\begin{figure}
\centering
  \includegraphics[width=1.0\columnwidth]{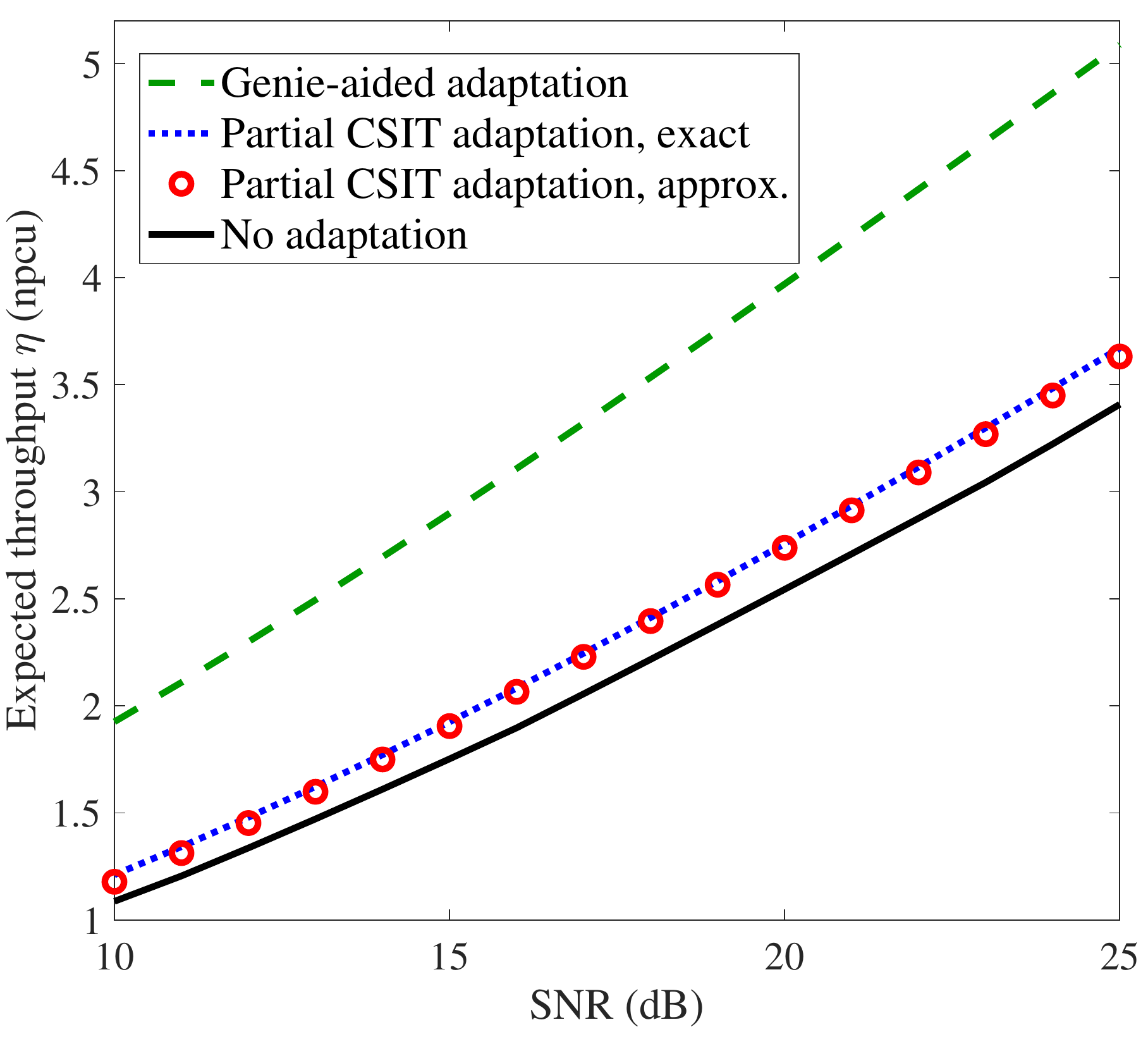}\\
\caption{Expected throughput $\eta$ in different cases, $v$ = 114 km/h, $\kappa$ = 1, and $\delta = $ 5 ms. Both the exact values estimated from  simulations as well as the analytical approximations from Lemma \ref{Lemma4} are presented.}\label{fig_Figure5}
\end{figure}
In this part, we study the performance of the PA system, and verify the tightness of the approximation scheme of Lemma \ref{Lemma4}. Particularly, we present the average throughput (\ref{eq_etaexpected}) and the outage probability of the PA setup for different vehicle speeds/channel estimation errors. As an ultimate upper bound for the proposed rate adaptation scheme, we consider a genie-aided setup where we assume that the BS has perfect CSIT of the BS-RA link without uncertainty/outage probability. Then, as a lower-bound of the system performance, we consider the cases with no CSIT/rate adaptation as shown in Fig. \ref{fig_Figure5}, i.e., $\sigma=1$ in (\ref{eq_H}). Here, the simulation results for the cases of no adaptation are obtained with only one antenna and no CSIT. In this case,  the data is sent with a fixed rate $R$ and it is decoded if $R<\log(1+gP)$, i.e., $g>\frac{e^{R}-1}{P}$. In this way, assuming  Rayleigh fading conditions, the average rate over all possible values of  $g$ is given by
\begin{align}\label{eq_bench}
    R^{\text{No-adaptation}} = \int_{\frac{e^{R}-1}{P}}^{\infty} Re^{-x} \text{d}x = Re^{-\frac{e^{R}-1}{P}},
\end{align}
and the optimal rate allocation is found by setting the derivative of (\ref{eq_bench}) with respect to $R$ equal to zero leading to $\tilde{R} = \mathcal{W}(P)$. Then, the throughput is calculated as 
\begin{align}\label{eq_etanocsi}
  \eta^{\text{No-adaptation}} =\mathcal{W}(P)e^{-\frac{e^{\mathcal{W}(P)}-1}{P}}.  
\end{align}
Also, in the simulations, we set  $f_\text{c}$ = 2.68 GHz in coherence with testbed results in, e.g., \cite{Jamaly2014EuCAP}, and $d_\text{a} = 1.5\lambda$ to avoid coupling effects. Finally, each point in the figures is obtained by averaging the system performance over $1\times10^5$ channel realizations.

\begin{figure}
\centering
  \includegraphics[width=1.0\columnwidth]{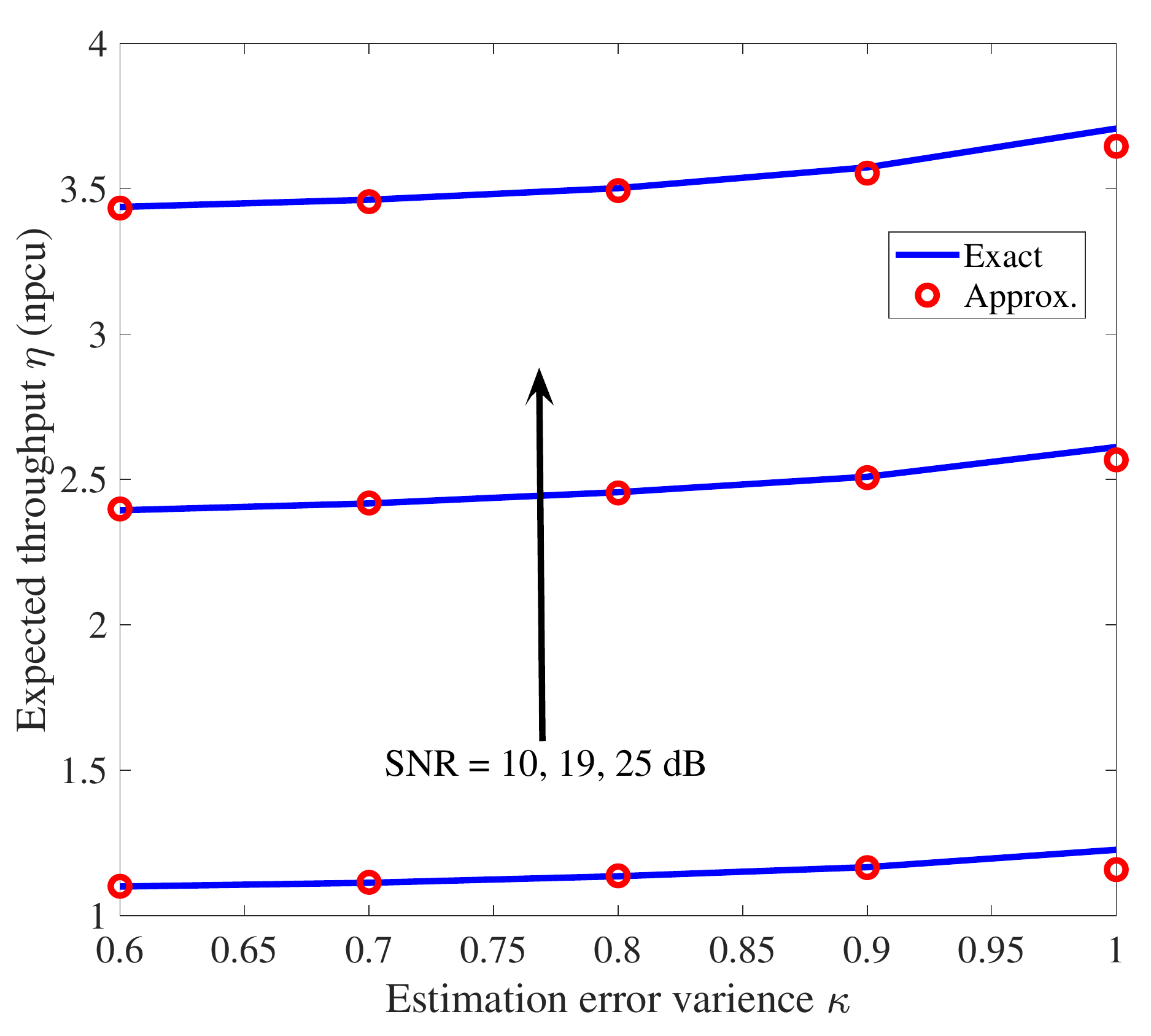}\\
\caption{Expected throughput $\eta$ for different estimation error variance $\kappa$ }with SNR = 10, 19, 25 dB, in the case of partial CSIT, exact and approximation,  $v = $  114.5 km/h, and $\delta = 5$ ms. Both the exact values estimated from simulations as well as the analytical approximations from Lemma \ref{Lemma4} are presented.\label{fig_Figure6}
\end{figure}

In Fig. \ref{fig_Figure5}, we show the expected throughput $\eta$  in different cases for a broad range of SNRs. Here, because the noise has unit variance, we define the SNR as $10\log_{10}P$ in Fig. \ref{fig_Figure5}. Also, we set $v = 114$ km/h which corresponds to mismatch distance $d\simeq0.1\lambda$.  The analytical results obtained by Lemma \ref{Lemma4} and Corollary \ref{coro2}, i.e., the approximation of (\ref{eq_optR}), are also presented. Note that, we have also verified the approximation result of Lemma \ref{Lemma4} while using Lemma \ref{Lemma1}/Corollary \ref{coro1}. Then, because the results are similar to those presented in Fig. \ref{fig_Figure5}, they are not included in the figure. Moreover, the figure shows the results of (\ref{eq_etanocsi}) with no CSIT/rate adaptation as a benchmark. Finally, Fig. \ref{fig_Figure6} studies the expected throughput $\eta$ for different values of estimation error variance $\kappa$ with SNR = 10, 19, 25 dB, in the cases with partial CSIT. Also, the figure evaluates the tightness of the approximation results obtained by Lemma \ref{Lemma4}. Here,  we set $v =$ 114.5 km/h and $\delta = 5$ ms.

Setting SNR = 23 dB and $v = $ 120, 150 km/h in Fig. \ref{fig_Figure7},  we study the effect of the processing delay $\delta$ on the throughput. Finally, the outage probability is evaluated in Fig. \ref{fig_Figure8}, where the results are presented for different speeds with SNR = 10 dB, in the cases with partial CSIT. Here, the mismatch distance $d$ is up to $0.25\lambda$. Also, we present  the outage probability for $\delta = 5.35$ ms and $\delta = 4.68$ ms in Fig.  \ref{fig_Figure8}.

\begin{figure}
\centering
  \includegraphics[width=1.0\columnwidth]{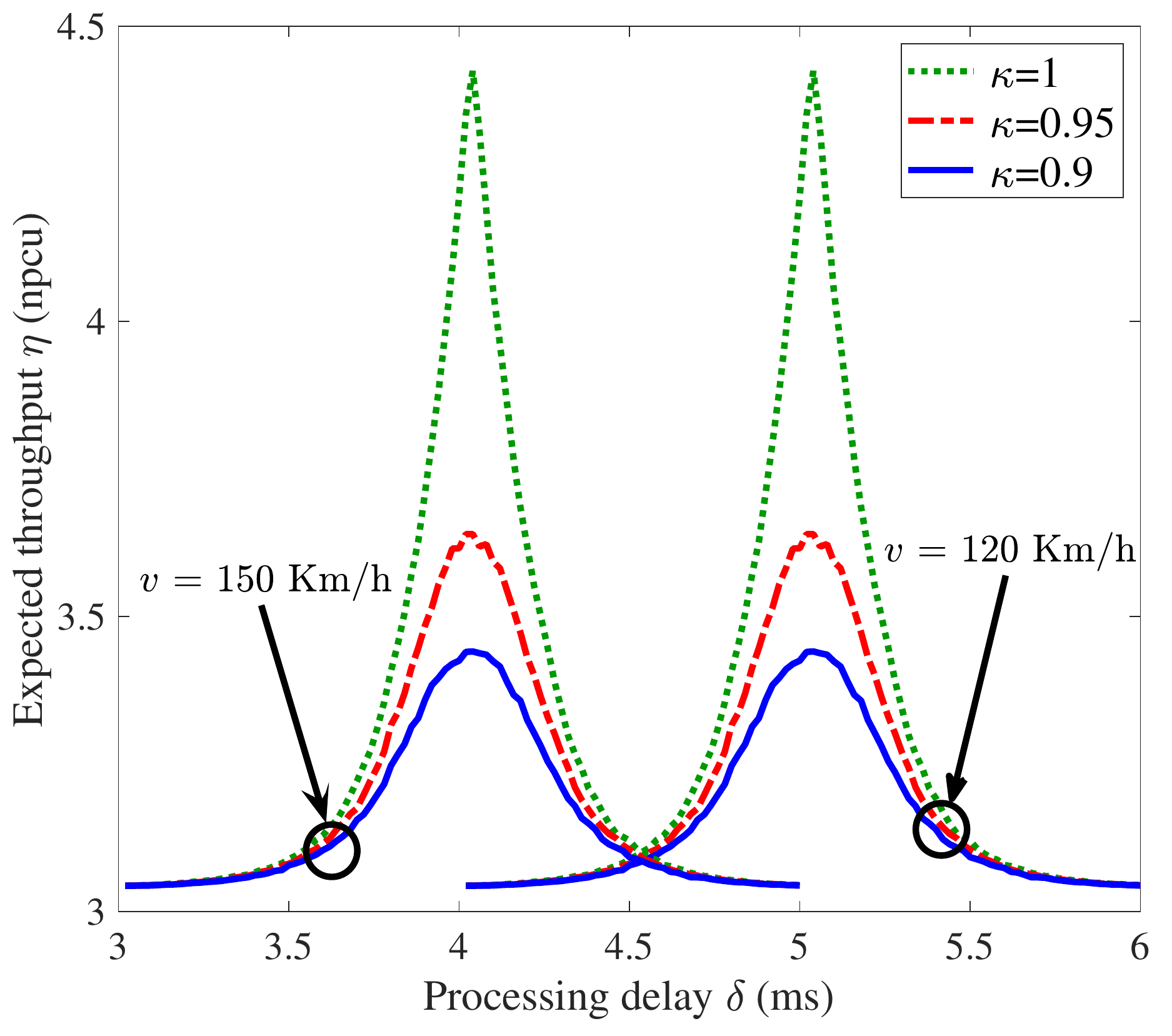}\\
\caption{Expected throughput $\eta$ for different processing delays with SNR = 23 dB and $v$ = 120, 150 km/h in the case of partial adaptation. }\label{fig_Figure7}

\end{figure}

\begin{figure}
\centering
  \includegraphics[width=1.0\columnwidth]{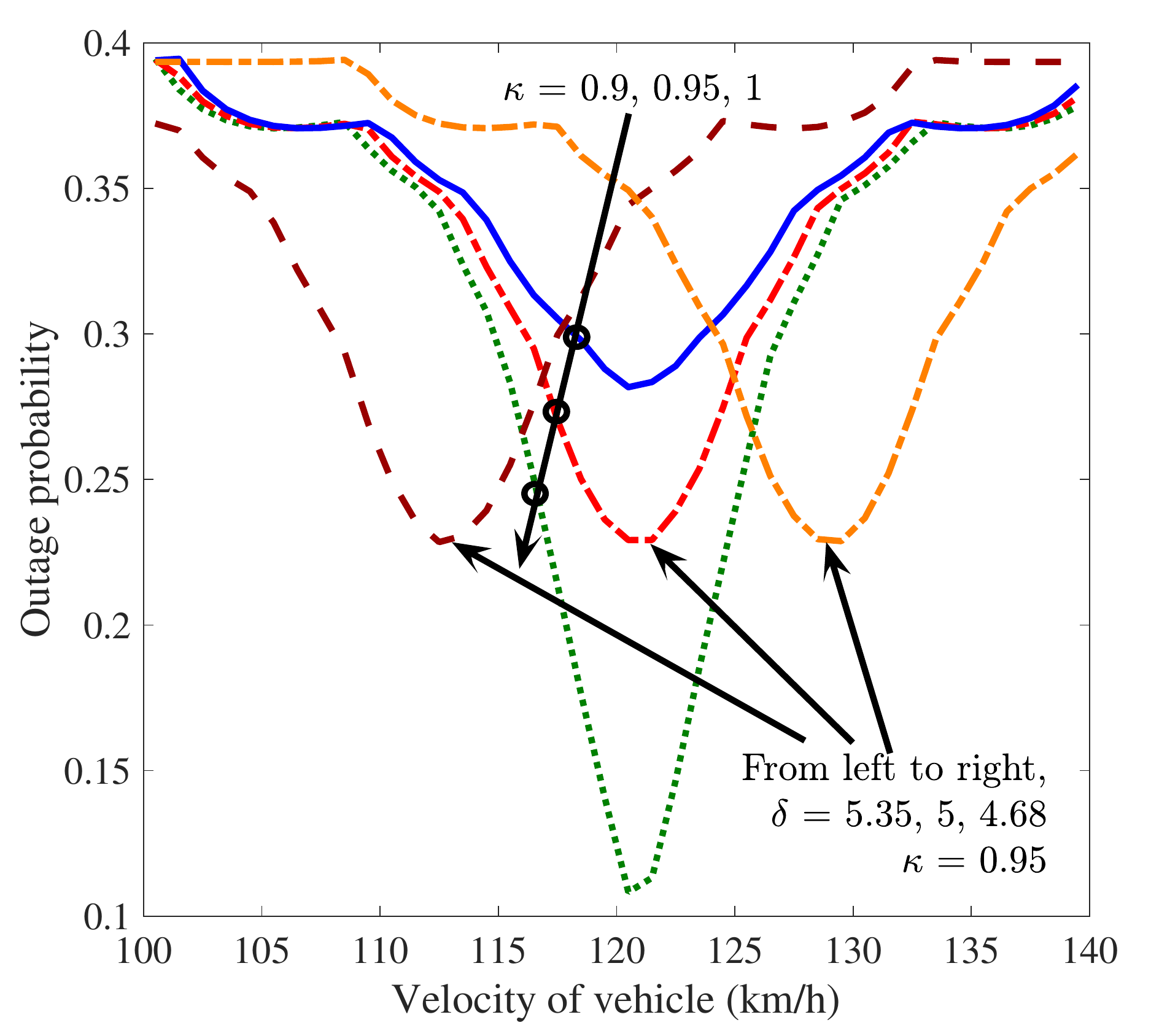}\\

\caption{Outage probability for different velocities with SNR = 10 dB, in the case of partial CSIT.}\label{fig_Figure8}

\end{figure}

From the figures, the following points can be concluded:
\begin{itemize}
    \item The approximation scheme of Lemma \ref{Lemma4} is tight for a broad range of parameter settings (Figs. \ref{fig_Figure5}, \ref{fig_Figure6}). Thus, the throughput-optimized rate allocation can be well approximated by (\ref{eq_appRF}), and the semi-linear approximation of Lemma \ref{Lemma1}/Corollaries \ref{coro1}-\ref{coro2} is a good approach to study the considered optimization problem.
    
    \item The deployment of the PA increases the throughput, compared to the open-loop setup, especially in moderate/high SNRs (Fig. \ref{fig_Figure5}). Also, the throughput decreases when the estimation error is considered, i.e., the variance $\kappa$ decreases. Finally, as can be seen in Figs.  \ref{fig_Figure5}, \ref{fig_Figure6}, with rate adaptation,  and without optimizing the processing delay/vehicle speed, the effect of estimation error on the expected throughput  is small for the cases with moderate mismatch  unless for large values of $\kappa$.

    \item As it can be seen in Figs. \ref{fig_Figure7} and \ref{fig_Figure8}, for different channel estimation errors, there are optimal values for the vehicle speed and the BS processing delay optimizing the system throughput and outage probability. Note that the presence of the optimal speed/processing delay can be proved via (\ref{eq_d}) as well. Finally, the optimal value of the vehicle speed, in terms of throughput/outage probability, decreases with the processing delay. However, the optimal vehicle speed/processing delay, in terms of throughput/outage probability, is almost insensitive to the channel estimation error.

    \item With perfect channel estimation, the throughput/outage probability is sensitive to the speed variation, if we move away from the optimal speed (Figs. \ref{fig_Figure7} and \ref{fig_Figure8}). That is, we can achieve significant benefits with the PA. However, to gain the full potential there is a higher requirement on the mismatch distance when $\kappa$ increases (Figs. \ref{fig_Figure7} and \ref{fig_Figure8}). Finally, considering Figs. \ref{fig_Figure7} and \ref{fig_Figure8}, it is expected that adapting the processing delay, as a function of the vehicle speed, will be an effective approach to improve the performance of PA setups. 
\end{itemize}

\section{Conclusion}\label{sec:conclu}
We derived a simple semi-linear approximation method for the first-order Marcum $Q$-function, as one of the functions of interest in different problem formulations of wireless networks. As we showed through various analysis, while the proposed approximation is not tight at the tails of the function, it is useful in different optimization- and expectation-based problem formulations. Particularly, as an application of interest, we used our proposed approximation to analyze the performance of PA setups using  rate adaptation. As we showed, with different levels of channel estimation error/processing delay, adaptive rate allocation can effectively compensate for the spatial mismatch problem, and improve the throughput/outage probability of PA networks.  It is expected that increasing the number of RA antennas will further improve the performance of the PA systems.



\appendices
\section{Proof of Lemma \ref{Lemma2}}
  \label{proof_Lemma2}
Using Corollary \ref{coro1},  we have
\begin{align}
G(\alpha,\rho)\simeq
\begin{cases}
\int_\rho^{\breve{c}_1}0\text{d}x + \int_{\breve{c}_1}^{\breve{c}_2}e^{-nx}x^{m}\times\nonumber\\
~~~\left(\frac{1}{\sqrt{2\pi}}(x-\alpha) + \frac{1}{2}\left(1-\frac{1}{\sqrt{2\pi\alpha^2}}\right)\right)\text{d}x +\nonumber\\
~~~\int_{\breve{c}_2}^{\infty}e^{-nx}x^{m}\text{d}x, ~~~~~~~~~~~~~~\mathrm{if}~  \rho < \breve{c}_1  \\ 
 \int_{\rho}^{\breve{c}_2}e^{-nx}x^{m}\times\nonumber\\
~~~\left(\frac{1}{\sqrt{2\pi}}(x-\alpha) + \frac{1}{2}\left(1-\frac{1}{\sqrt{2\pi\alpha^2}}\right)\right)\text{d}x +\nonumber\\
~~~\int_{\breve{c}_2}^{\infty}e^{-nx}x^{m}\text{d}x, ~~~~~~~~~~~~~~\mathrm{if}~  \breve{c}_1\leq\rho<\breve{c}_2,\nonumber\\
\int_\rho^{\infty} e^{-nx}x^{m}\text{d}x, ~~~~~~~~~~~~~~~~~~\mathrm{if}~  \rho\geq\breve{c}_2.
\end{cases}
\end{align}

Then, for $\rho\geq \breve{c}_2$, we obtain
\begin{align}
&    \int_\rho^\infty e^{-nx}\times x^{m}(1-Q_1(\alpha,x))\text{d}x\nonumber\\
  &  \overset{(d)}\simeq\int_\rho^\infty e^{-nx}\times x^{m}\text{d}x
   \overset{(e)}= \Gamma(m+1,n\rho)n^{-m-1}, 
\end{align}
while for $\rho<\breve{c}_2$, we have

\begin{align}
 &  \int_\rho^\infty e^{-nx}\times x^{m}(1-Q_1(\alpha,x))\text{d}x\nonumber\\
&    \overset{(f)}\simeq \int_{\max(\breve{c}_1,\rho)}^{\breve{c}_2} \left(\frac{1}{\sqrt{2\pi}}(x-\alpha) + 
\frac{1}{2}\left(1-\frac{1}{\sqrt{2\pi\alpha^2}}\right)\right)\times\nonumber\\
& ~~~~~~e^{-nx}x^{m}\text{d}x +\int_{\breve{c}_2}^\infty e^{-nx} x^{m}\text{d}x\nonumber\\
& \overset{(g)}= \Gamma(m+1,n\breve{c}_2)n^{-m-1} + \nonumber\\
&~~~~~~\left(-\frac{\alpha}{\sqrt{2\pi}}+0.5*\left(1-\frac{1}{\sqrt{2\pi\alpha^2}}\right)\right)\times n^{-m-1}\times\nonumber\\
&~~~~~~\left(\Gamma\left(m+1,n\max(\breve{c}_1,\rho)\right)-\Gamma\left(m+1,n\breve{c}_2\right)\right)+\nonumber\\
&~~~~~~\left(\Gamma(m+2,n\max(\breve{c}_1,\rho))-\Gamma(m+2,n\breve{c}_2)\right)\frac{n^{-m-2}}{\sqrt{2\pi}}.
\end{align}
Note that $(d)$ and $(f)$ come from Corollary \ref{coro1} while $(e)$ and $(g)$ use the fact that $\Gamma(s,x)\rightarrow 0$ as $x\rightarrow\infty$.

\section{Proof of Lemma \ref{Lemma3}}
\label{proof_Lemma3}
Using Lemma \ref{Lemma1}, the  integral (\ref{eq_integralT})  can be approximated as

1) for $\theta_2>\theta_1>c_2$, $T(\alpha,m,a,\theta_1,\theta_2) \simeq 0$,

2) for $c_1<\theta_1\leq c_2$,
\begin{align}
  &  T(\alpha,m,a,\theta_1,\theta_2) = \nonumber\\& ~~\int_{\theta_1}^{\min(c_2,\theta_2)} (n_2x+n_1)e^{-mx}\log(1+ax) \text{d}x,
\end{align}

3) for $\theta_1<c_1,c_2>\theta_2>c_1$,
\begin{align}
   & T(\alpha,m,a,\theta_1,\theta_2) \simeq \int_{\theta_1}^{c_1} e^{-mx}\log(1+ax)\text{d}x + \nonumber\\
   &~~ \int_{c_1}^{\theta_2} (n_2x+n_1)e^{-mx}\log(1+ax) \text{d}x,
\end{align}

4) for $\theta_1<\theta_2<c_1$,
\begin{align}
   & T(\alpha,m,a,\theta_1,\theta_2) \simeq \int_{\theta_1}^{\theta_2} e^{-mx}\log(1+ax)\text{d}x.
\end{align}

Then, consider case 3 where
\begin{align}
    T(\alpha,m,a,\theta_1,\theta_2) &= \int_{\theta_1}^{\theta_2} e^{-mx}\log(1+ax)Q_1(\alpha,x)\text{d}x\nonumber\\
    &\simeq \int_{\theta_1}^{c_1} e^{-mx}\log(1+ax)\text{d}x + \nonumber\\
   & \int_{c_1}^{\theta_2} (n_2x+n_1)e^{-mx}\log(1+ax) \text{d}x\nonumber\\
   & =  \mathcal{F}_1(c_1) -  \mathcal{F}_1(\theta_1) +  \mathcal{F}_2(\theta_2) -  \mathcal{F}_2(c_1),
\end{align}
with $n_1$ and $n_2$ being functions of the constant $\alpha$ and given by (\ref{eq_n1}) and (\ref{eq_n2}), respectively. The other cases can be proved with the same procedure. Moreover, the functions $\mathcal{F}_1(x)$ and $\mathcal{F}_2(x)$ are obtained by
\begin{align}
    \mathcal{F}_1(x) &= \int e^{-mx}\log(1+ax)\text{d}x\nonumber\\
    &\overset{(h)}= -\frac{e^{-mx}\log(ax+1)}{m}-\int -\frac{ae^{-mx}}{m(ax+1)} \text{d}x\nonumber\\
    & \overset{(i)}=  \frac{1}{m}\left(-e^{\frac{m}{a}}\operatorname{E_1}\left(mx+\frac{m}{a}\right)-e^{-mx}\log(ax+1)\right)+C,
\end{align}
and
\begin{align}
     \mathcal{F}_2(x)& =  \int (n_2x+n_1)e^{-mx}\log(1+ax) \text{d}x\nonumber\\
    &\overset{(j)} = -\frac{(mn_2x+n_2+mn_1)e^{-mx}\log(1+ax)}{m^2} - \nonumber\\
    &~~~\int \frac{a(-mn_2x-n2-mn_1)e^{-mx}}{m^2(ax+1)} \text{d}x\nonumber\\
    & \overset{(k)}= -\frac{(mn_2x+n_2+mn_1)e^{-mx}\log(1+ax)}{m^2} - \nonumber\\
    &~~~ -\frac{1}{a}\bigg(e^{\frac{m}{a}}(mn_2x+n_2+mn_1)\nonumber\\
    &~~~ \operatorname{E_1}\left(mx+\frac{m}{a}\right)\bigg)+
    \int - \frac{me^{\frac{m}{a}}n_2}{\operatorname{E_1}\left(mx+\frac{m}{a}\right)} \text{d}x\nonumber\\
    & \overset{(l)}= -\frac{(mn_2x+n_2+mn_1)e^{-mx}\log(1+ax)}{m^2} - \nonumber\\
    &~~~ -\frac{1}{a}\bigg(e^{\frac{m}{a}}(mn_2x+n_2+mn_1) \operatorname{E_1}\left(mx+\frac{m}{a}\right)\bigg)+\nonumber\\
    &~~~  \frac{n_2e^{-mx}}{a} - \frac{e^{\frac{m}{a}}n_2(mx+\frac{m}{a}) \operatorname{E_1}\left(mx+\frac{m}{a}\right)}{a}+C,
\end{align}
where $(h)$, $(j)$ and $(k)$ come from partial integration and some manipulations. Also, $(i)$ and $(l)$ use \cite[p. 195]{geller1969table}
\begin{align}
    \int \operatorname{E_1}(u) \text{d}u = u\operatorname{E_1}(u)-e^{-u},
\end{align}
with $\operatorname{E_1}(x) = \int_x^{\infty} \frac{e^{-t}}{t} \mathrm{d}t$ being the Exponential Integral function \cite[p. 228, (5.1.1)]{abramowitz1999ia}.

\balance
\bibliographystyle{IEEEtran}
\bibliography{main}

\end{document}